\documentclass[times,sort&compress,3p]{article}
\usepackage{geometry}
\usepackage[labelfont=bf]{caption}
\usepackage{amsmath,amsfonts,amssymb,amsthm,booktabs,color,epsfig,graphicx,url}
\usepackage{bm}
\usepackage{verbatim}
\usepackage{graphicx}
\usepackage{natbib}
\usepackage{euscript,bbm,xcolor}
\usepackage{float}
\usepackage[export]{adjustbox}
\usepackage{caption}
\usepackage{subcaption}

\RequirePackage[colorlinks,citecolor=blue,urlcolor=blue]{hyperref}

\newcommand{\R}{\mathbb{R}}

\newcommand{\diag}{{\rm diag}}

\renewcommand{\P}{{\mathbb{P}}}
\newcommand{\E}{{\mathbb{E}}}
\newcommand{\Var}{{\rm Var}}

\newcommand{\Cor}{{\rm Cor}}

\newcommand{\adj}{{\rm adj}}

\newcommand{\GOI}{{\rm GOI}}

\newcommand{\Rad}{{\rm Rad}}

\newcommand\numberthis{\addtocounter{equation}{1}\tag{\theequation}}

\newtheorem{thm}{Theorem}
\newtheorem{lemma}{Lemma}
\newtheorem{prop}{Proposition}
\newtheorem{assumption}{Assumption}

\theoremstyle{plain}

\begin{document}

\title{An approximation to peak detection power using Gaussian random field theory}

\author{Yu Zhao$^1$ \and Dan Cheng$^3$ \and Armin Schwartzman$^{1,2}$}
\date{%
    $^1$Division of Biostatistics, \\ Herbert Wertheim School of Public Health and Human Longevity Science, \\
University of California San Diego, 9500 Gilman Dr., La Jolla, CA 92093, USA \\%
    $^2$Halicioǧlu Data Science Institute, \\
University of California San Diego, 9500 Gilman Dr., La Jolla, CA 92093, USA \\%
    $^3$School of Mathematical and Statistical Sciences, \\ Arizona State University, 900 S. Palm Walk, Tempe, AZ 85281, USA \\[2ex]%
}
\maketitle

\begin{abstract}
We study power approximation formulas for peak detection using Gaussian random field theory. The approximation, based on the expected number of local maxima above the threshold $u$, $\E[M_u]$, is proved to work well under three asymptotic scenarios: small domain, large threshold, and sharp signal. An adjusted version of $\E[M_u]$ is also proposed to improve accuracy when the expected number of local maxima $\E[M_{-\infty}]$ exceeds 1. 

\citet{Bernoulli} developed explicit formulas for $\E[M_u]$ of smooth isotropic Gaussian random fields with zero mean. In this paper, these formulas are extended to allow for rotational symmetric mean functions, so that they are suitable for power calculations. We also apply our formulas to 2D and 3D simulated datasets, and the 3D data is induced by a group analysis of fMRI data from the Human Connectome Project to measure performance in a realistic setting.
\end{abstract}

\begin{keywords} 
Power calculations, peak detection, Gaussian random field, image analysis.
\end{keywords}

\section{Introduction}

Detection of peaks (local maxima) is an important topic in image analysis. For example, a fundamental goal in fMRI analysis is to identify the local hotspots of brain activity (see, for example, \Citealp{Genovese2002} and \Citealp{Heller2006}), which are typically captured by peaks in the fMRI signal. The detection of such peaks can be posed as a statistical testing problem intended to test whether the underlying signal has a peak at a given location. This is challenging because such tests are conducted only at locations of observed peaks, which depend on the data. Therefore, the height distribution of the observed peak is conditional on a peak being observed at that location. This is a nonstandard problem. Solutions exist using random field theory (RFT). RFT is a statistical framework that can be used to perform topological inference and modeling. RFT-based peak detection has been studied in \citet{Annals} and \Citet{Telschow19}, which provide the peak height distribution for isotropic noise under the complete null hypothesis of no signal anywhere.

In general, for any statistical testing problem, accurate power calculations help researchers decide the minimum sample size required for an informative test, and thus reduce cost. Power calculation formulas exist for common univariate tests, such as z-tests and t-tests. However, particular challenges arise when we perform power calculations in peak detection settings. Due to the nature of imaging data, the number and location of the signal peaks are unknown. Besides, the power is affected by other spatial aspects of the problem, such as the shape of the peak function and the spatial autocorrelation of the noise. Considering these difficulties, it requires some extra effort to derive a power formula for peak detection.

A formal definition of power in peak detection is necessary to perform power calculations. In \citet{Annals} and \citet{Durnez16}, the authors explored approaches to control the false discovery rate (FDR). For the entire domain, average peakwise power, i.e. power averaged over all non-null voxels, is a natural choice for these approaches. For a local domain where a single peak exists, the power can be defined as the probability of successfully detecting that peak. Following this idea, we describe the null and alternative hypothesis and the definition of detection power. We do so informally here for didactic purposes and present formal definitions in Section \ref{sec:power_approx}.

Consider a local domain where a single peak may exist, and consider the hypotheses
\begin{align*}
H_0 &: \text{``the signal is equal to 0 in the local domain." \quad vs} \\
H_1 &: \text{``the signal has at least one positive peak in the local domain."}
\end{align*}
 Suppose we observe a random field to be used as test statistic at every location, typically as the result of statistical modeling of the data. For a fixed threshold $u$, the existence of observed peaks with height greater than $u$ would lead to rejecting the null hypothesis. Therefore, we define the type I error and power as the probability of existing at least one local maximum above $u$ under $H_0$ and $H_1$ respectively:
\begin{align*}
   \text{Type I error: }& \mathbb{P}\{\exists \text{ a peak in the local domain with height} >  u \text{ when $H_0$ is true}\} \\
   \text{Power: }& \mathbb{P}\{\exists \text{ a peak in the local domain with height} >  u \text{ when $H_1$ is true}\}
\label{eqn:power_def}
\numberthis
\end{align*}
\noindent
Formulas for type I error have been developed for stationary fields in 1D and isotropic fields in 2D and 3D (\citealp{Extreme}, \citealp{Annals} and \citealp{Bernoulli}). However, there is no formula to calculate power. In order to get an appropriate estimate of power, we need to know the peak height distribution for non-centered (the mean function is not 0) random fields. Generally speaking, it is very difficult to calculate the peak height distribution especially when the random field has non-zero mean. \citet{Durnez16} suggests using Gaussian distribution to describe the non-null peaks and truncated Gaussian distribution to approximate the overshoot distribution. This approach is easy to implement but not very accurate because the peak height distribution is in reality always skewed and not close to any Gaussian distribution. 

In this article, we propose to approximate the probability of an observed peak exceeding the detection threshold $u$ by calculating the expected number of peaks above the threshold $u$. We show that the approximation, which is also an upper bound, works well under certain scenarios. For the entire domain, we can approximate the average peakwise power by taking the arithmetic mean of the approximation proposed in this paper over non-null voxels.

The proposed approximation makes the problem more tractable, but in general, it does not have an explicit form. In order to make it applicable in practice, we further simplify the formula under the isotropy assumption and show its explicit form in 1D, 2D, and 3D. The explicit results are validated through 2D and 3D computer simulations carried out in MATLAB. The simulation also covers multiple scenarios by modifying the parameters used to generate the data. The performance of power approximation and its conservative adjustment under these scenarios are discussed. 

Finally, to assess the real-data performance of our power approximation method, we apply it to a 3D simulation induced by a real brain imaging dataset, where the parameters are estimated from the Human Connectome Project (\citealp{HCP2012}) fMRI data. By testing the method in a realistic setting, we also demonstrate how effect size and other parameters affect the power.

The paper is organized as follows. We first show in Section \ref{sec:power_approx} the problem setup and theoretical results in certain scenarios. In Section \ref{sec:explicit}, we derive the explicit formulas under isotropy. Simulation in 2D is conducted in Section \ref{sec:sim}. Details regarding how to apply our formula in application setting are discussed in Section \ref{sec:application}. The methodology is applied to a 3D real dataset in Section \ref{sec:3d_example}. 

\section{Power Approximation}\label{sec:power_approx}

\subsection{Setup}

Let $Y(s) = \sigma(s) Z(s)+\mu(s)$ where $Z=\{Z(s),s\in D\}$ representing the noise is a centered (zero-mean) smooth unit-variance Gaussian random field on an $N$ dimensional non-empty domain $D \subset \mathbb{R}^N$, $\sigma(s)$ is the standard deviation of the noise and $\mu(s)$ is the mean function. Let $X(s) = Y(s)/\sigma(s) = Z(s) + \theta(s)$ where the ratio $\theta(s) = \mu(s)/\sigma(s)$ is the standardized mean function, which we assume to be $C^2$. Here $C^3$ is a sufficient smoothness condition for $Z$, and this will be clarified in Assumption \ref{con:1} below. 

Let
\begin{equation*}
\begin{split}
X_i(s)&=\frac{\partial X(s)}{\partial s_i}, \quad \nabla X(s)= (X_1(s), \ldots, X_N(s)),\\
X_{ij}(s)&=\frac{\partial^2 X(s)}{\partial s_is_j}, \quad \nabla^2 X(s)= (X_{ij}(s))_{1\le i, j\le N},\\
Z_i(s)&=\frac{\partial Z(s)}{\partial s_i}, \quad \nabla Z(s)= (Z_1(s), \ldots, Z_N(s)),\\
Z_{ij}(s)&=\frac{\partial^2 Z(s)}{\partial s_is_j}, \quad \nabla^2 Z(s)= (Z_{ij}(s))_{1\le i, j\le N}.
\end{split}
\end{equation*}

We will make use of the following assumptions:

\begin{assumption}\label{con:1}

$Z \in C^2(D)$ almost surely and its second derivatives satisfy the mean-square H\"older condition: for any $s_0 \in D$, there exists positive constants $L$, $\eta$ and $\delta$ such that 
\begin{equation*}
\mathbb{E}(Z_{ij}(s)-Z_{ij}(t))^2 \leq L^2 \|s-t\|^{2\eta}, \quad \forall t,s \in U_{s_0}(\delta),\: i, j = 1,...,N.
\end{equation*}
where $U_{s_0}(\delta) = s_0 \oplus (-\delta/2,\delta/2)^N$ is the $N$ dimensional open cube of side length $\delta$ centered at $s_0$. This condition is satisfied, for example, if $Z$ is $C^3(D)$.

\end{assumption}

\begin{assumption}\label{con:2}
For every pair $(t,s) \in D \times D$ with $s \neq t$, the Gaussian random vector
\begin{equation*}
(Z(s),\,\nabla Z(s),\, Z_{ij}(s),\, Z(t),\, \nabla Z(t),\, Z_{ij}(t),\, 1\leq i \leq j \leq N)
\end{equation*}
is non-degenerate, i.e. its covariance matrix has full rank.

\end{assumption}

\subsection{Peak detection}

Following the notation in the problem setup, the null and alternative hypothesis can be written as:
\begin{align*}
H_0 &: \text{$\mu(s) = 0 \: \text{for all} \: s \in D$ \quad vs} \\
H_1 &: \text{$\mu(s) > 0, \: \nabla \mu(s) = 0 \text{, }\nabla^2\mu(s) \prec 0 \: \text{for some } s \in D$}
\end{align*}
The mean function $\mu(s)$ is not directly observed, so the hypothesis is tested based on the peak height of $X(s)$. For a peak detection procedure that aims to test this hypothesis, a threshold $u$ for the peak height of $X(s)$ needs to be set in advance. If a local maximum with height greater than $u$ is observed, we would choose to reject the null hypothesis due to the strong evidence against it. The probability that a peak of $X$ exceeds $u$

 \begin{equation}
 \P{\left(\exists\: s \in D \: s.t. \: X(s)>u| \nabla X(s) = 0 \text{ and }\nabla^2X(s) \prec 0\right)}
 \label{def:power_X}
 \end{equation}
is the type I error under $H_0$ and power under $H_1$. The threshold $u$ can be obtained based on the peak height distribution under $H_0$. A formula for peak height distribution of smooth isotropic Gaussian random fields has been derived in \citet{Bernoulli} and it can also be derived directly from a special case of the formulas presented in this paper. Usually, $u$ is set to be some quantile of the null distribution of peak height to maintain the nominal $\alpha$ type I error. More details about selecting the threshold will be discussed in the real data example. Selecting $u$ is not the main focus of this paper and our method can be applied to any choice of $u$.

\subsection{Power approximation}

Let $M_u$ be the number of local maxima of the random field $X$ above $u$ over the local domain $D$.  The power defined in \eqref{def:power_X} can be represented as $\P[M_u \ge 1]$. We call this the {\em power function}, seen as a function of the threshold u. Note that 
\begin{equation}
\P[M_u \ge 1] = \sum_{k=1}^\infty \P[M_u = k] \le \sum_{k=1}^\infty k\P[M_u = k] = \E[M_u].
\end{equation}
On the other hand,
\begin{equation}
\E[M_u]-\P[M_u \ge 1] = \sum_{k=2}^\infty (k-1)\P[M_u = k] \le \frac{1}{2}\sum_{k=2}^\infty k(k-1)\P[M_u = k] = \frac{1}{2}\E[M_u(M_u-1)].
\end{equation}
Thus, we have
\begin{equation}\label{ieq:power_approx}
\E[M_u]-\frac{1}{2}\E[M_u(M_u-1)] \le \P[M_u \ge 1] \le \E[M_u].
\end{equation}
This inequality tells us that for any fixed $u$, the power is bounded within an interval of length $ \E[M_u(M_u-1)]/2$. Thus, $\E{[M_u]}$ is a good approximation of power if one of the two conditions below is satisfied:
\begin{enumerate}
\item The factorial moment $\E{[M_u(M_u-1)]}$ converges to 0 and $\E{[M_u]}$ does not.
\item They both converge to 0 and $\E{[M_u(M_u-1)]}$ converges faster than $\E{[M_u]}$.
\end{enumerate} 
The convergence above refers to conditions on the signal and noise parameters. In the rest of this section, we introduce four interesting results. The first result can be useful for simplifying the power function and the other three results give different scenarios where one of the conditions above holds. 

\subsection{Adjusted $\E[M_u]$}

We have provided evidence of using $\E[M_u]$ to approximate power through \eqref{ieq:power_approx}. However, $\E[M_u]$ alone might not be sufficient for power approximation since it only gives an upper bound. Also, unlike power, $\E[M_u]$ sometimes exceeds 1. To correct for this, we define the adjusted $\E[M_u]$ as
\begin{equation}\label{eqn:Mu_adj}
\E[M_u]_{\adj} = \E[M_u]/\max(1,\E[M_{-\infty}]).
\end{equation}
The adjusted $\E[M_u]$ is the same as $\E[M_u]$ when the expected number of local maxima $\E[M_{-\infty}]$ is less or equal to 1. When $\E[M_{-\infty}]$ is greater than 1, we divide $\E[M_u]$ by $\E[M_{-\infty}]$ to make sure it never exceeds 1. The adjusted $\E[M_u]$ is more conservative, and we conjecture that it is a lower bound of power when there exists at least one local maximum in the domain $D$. In applications, people are interested in a conservative estimator so that the test is guaranteed to have enough power. Combining $\E[M_u]$ and $\E[M_u]/\E[M_{-\infty}]$, we can get an approximate range of the true power. We will compare $\E[M_u]$ and adjusted $\E[M_u]$ in simulation studies. 

\subsection{Height equivariance}

Our first result does not concern the approximation \eqref{ieq:power_approx} yet, but it offers a simplification of the power function and $\E[M_u]$ that will be used later. The proposition below states that the power function and $\E[M_u]$ for peak detection are translation equivariant with respect to peak height.

\begin{prop}\label{prop:height_equi}
Let $\theta(s) = h(s) + \theta_0$ be a peak signal with height $\theta_0$, where $h(s)$ is a unimodal mean function with maximum equal to 0 at $s_0$ in $D$. Then the power function for peak detection and $\E[M_u]$ can be written in the form $F(u-\theta_0)$, where $F(u)$ is the power function or $\E[M_u]$ at $\theta_0 = 0$.
\end{prop}

\begin{proof}
Let $\tilde{\theta}(s) = \theta(s) - \theta_0 = h(s) + 0$ and $\tilde{M}_u$ be the number of local maxima of the random field $\tilde{X}(s) = Z(s) + \tilde{\theta}(s)$ above $u$ over $D$. Considering the definition of power, we have
\begin{equation*}
F(u-\theta_0) = \P[\tilde{M}_{u-\theta_0} \geq 1] = \P[M_u\geq 1].
\end{equation*}
Given that $\E[\tilde{M}_{u-\theta_0} ] = \E[M_u]$, is is also straightforward to show $\E[M_u]$ is translation equivariant with respect to $\theta_0$.
\end{proof}

Next, we give three scenarios where the equality in (\ref{ieq:power_approx}) can be achieved asymptotically: small domain size, large threshold, and sharp signal. 

\subsection{Small domain}

If the size of the local domain $D$ where a single peak exists is small enough, it can be shown that equality in (\ref{ieq:power_approx}) can be achieved asymptotically. 

\begin{thm}
Consider a local domain $D_{\epsilon}= U(s_0,\epsilon)$ for any fixed $s_0 \in D$ where $U(s_0,\epsilon)= t_0 \oplus (-\epsilon/2, \epsilon/2)^N$ is the N-dimensional open cube of side $\epsilon$ centered at $t_0$. For sufficiently small $\epsilon$ and fixed threshold u,
\begin{equation}\label{eqn:small_domain}
\P[M_u \geq 1]=\E[M_u](1-o(1)) = \E[M_u]_{\adj}(1-o(1)),
\end{equation}
\label{thm:small_domain}
\end{thm}

\begin{proof}
The proof is based on the proof of Lemma 3 in \citet{Piterbarg96} and Lemma 4.1 in \citet{Extreme}.
\begin{equation}
\begin{split}
\E[M_u(M_u-1)] = \int_{D_{\epsilon}}\int_{D_{\epsilon}} \int_u^{\infty}\int_u^{\infty} \E\left[|{\rm det}\nabla^2X(s)||{\rm det}\nabla^2X(t)|\middle|\begin{matrix}
X(s)=x_1,X(t)=x_2 \\
\nabla X(s)=\nabla X(t) = 0
\end{matrix}\right] \\ 
P_{X(s),X(t),\nabla X(s),\nabla X(t)}(x_1,x_2,0,0) \, dx_1 \, dx_2 \, ds \, dt.
\end{split}
\label{eqn:factorial_moment}
\end{equation}

Let 
\begin{equation*}
\E_1(s,t)= \E\left[|{\rm det}\nabla^2X(s)||{\rm det}\nabla^2X(t)|\middle|\begin{matrix}
X(s)=x \\
\nabla X(s)=\nabla X(t) = 0
\end{matrix}\right]    
\end{equation*}
and replace one of the integration limits in (\ref{eqn:factorial_moment}) by $-\infty$, we have
\begin{equation*}
\begin{split}
\E[M_u(M_u-1)] \leq \int_{D_{\epsilon}} \int_{D_{\epsilon}} P_{\nabla X(s),\nabla X(t)}(0,0) ds dt \int_u^{\infty}  \E_1(s,t) P_{X(s)}\left(x \middle|\nabla X(s)=\nabla X(t) =0 \right) dx.
\end{split}
\end{equation*}

Then we can take the Taylor expansion
\begin{equation*}
\nabla X(t) = \nabla X(s) + \nabla^2 X(s) (t-s)+ ||t-s||^{1+\alpha} Y_{s,t}
\end{equation*}
where $Y_{s,t} = (Y_{s,t}^1,...,Y_{s,t}^N)^T$ is a Gaussian vector field. Note that the determinant of $\nabla^2 X(s)$ is equal to the determinant of 

\begin{equation}
\begin{pmatrix}
1 & -(t_1-s_1) & \hdots & -(t_N-s_N) \\
0 & & &\\
\vdots & & \nabla^2 X(s) &\\
0 & & &
\end{pmatrix}
\label{mat:det}
\end{equation}

For $i=2,...,N+1$, multiply the $i$th column of this matrix by $(t_i-s_i)/||t_i-s_i||^2$, take the sum of all such columns and add the result to the first column. Since $\nabla X(s) = \nabla X(t) = 0$, we can derive $\nabla^2 X(s) (t-s) =-||t-s||^{1+\alpha}Y_{s,t}$, and obtain the matrix below with the same determinant as (\ref{mat:det})

\begin{equation*}
\begin{pmatrix}
0 & -(t_1-s_1) & \hdots & -(t_N-s_N) \\
-||t-s||^{-1+\alpha}Y_{s,t}^1 & & &\\
\vdots & & \nabla^2 X(s) &\\
-||s-t||^{-1+\alpha}Y_{s,t}^N & & &
\end{pmatrix}
\label{mat:det_2}
\end{equation*}

Let $r=\max_{1\leq i \leq N}|t_i-s_i|$, 
\begin{equation*}
A_{s,t} = \begin{pmatrix}
0 & -(t_1-s_1)/r & \hdots & -(t_N-s_N)/r \\
Y_{s,t}^1 & & &\\
\vdots & & \nabla^2 X(s) &\\
Y_{s,t}^N & & &
\end{pmatrix}
\label{mat:Ats}
\end{equation*}
So we have  
\begin{equation*}
\E_1(s,t) \leq ||t-s||^{\alpha}\E_2(s,t)
\end{equation*}
where 
\begin{equation*}
\E_2(s,t) = \E\left[ |{\rm det}A_{s,t}||{\rm det}\nabla^2X(t)| \middle| \begin{matrix}
X(s) = x,\nabla X(s) = 0 \\
\nabla^2 X(s)(t-s) = -||t-s||^{1+\alpha}Y_{s,t}
\end{matrix}\right].
\end{equation*}

Using the inequality of arithmetic and geometric means, we can bound the determinant 
\begin{equation*}
\begin{split}
|{\rm det}\nabla^2X(t)| \leq N^{2N-2} \sum_{i,j}|X_{ij}(t)|^N \\
|{\rm det} A_{s,t}| \leq (N+1)^{2N}\sum_{i,j}|a_{ij}|^{N+1}
\end{split}
\end{equation*}
where $a_{ij}$ is the $i,j$ entry of $A_{s,t}$. Apply the inequality again

\begin{equation*}
|{\rm det} \nabla^2 X(t)||{\rm det}A_{s,t}| \leq \frac{1}{2}N^{2N-2}(N+1)^{2N+1}\left(\sum_{i,j} |X_{ij}(t)|^{2N}+\sum_{i,j}|a_{ij}|^{2N+2} \right).
\end{equation*}

For any Gaussian variable $X$ and integer $N \geq 0$, the following inequality holds

\begin{equation*}
\E[X^{2N}] \leq 2^{2N}(\E[X]^{2N}+C_N \Var(X)^{2N})
\end{equation*}
where $C_N$ is a constant depending on $N$. Next, we can focus on the conditional expectation and conditional variance of $X_{ij}(t)$ and $Y_{s,t}$.

By Assumption \ref{con:1} and \ref{con:2} and the fact that the conditional variance of a Gaussian variable is less or equal to the unconditional variance, we can conclude that the conditional variance of $X_{ij}(t)$ and $Y_{s,t}$ are bounded above by some constant.


Summarizing the results above,
\begin{equation*}
\sup_{s,t \in D_{\epsilon}, s\neq t}|\E_2(s,t)| \leq C_1
\end{equation*}
for some constant $C_1>0$ and 
\begin{equation*}
\E_1(s,t) \leq ||t-s||^\alpha \E_2(s,t) \leq C_1||t-s||^\alpha.
\end{equation*}
Combine the results above and with a fixed threshold $u$ 
\begin{align*}
& \int_u^{\infty}  \E_1(s,t) P_{X(s)}\left(x \middle|\nabla X(s)=\nabla X(t) =0 \right) dx \\
& \leq C_1||t-s||^\alpha \int_u^{\infty}P_{X(s)}\left(x \middle|\nabla X(s)=\nabla X(t) =0 \right) dx \\
& =  C_1||t-s||^\alpha \int_u^{\infty}\exp(-(Ax-B)^2) dx \quad \text{for some constant $A$, $B$}\\
& = C_2||t-s||^\alpha
\end{align*}
for some constant $C_2>0$.

Next, by the proof of Lemma 4.1 in \cite{Extreme}

\begin{equation*}
p_{\nabla X(s),\nabla X(t)}(0,0) \leq C_3 ||t-s||^{-N}
\end{equation*}
for some constant $C_3>0$.

Therefore, there exists $C_4$ such that 
\begin{equation*}
\begin{split}
\E[M_u(M_u-1)] \leq C_4\int_{D_{\epsilon}} \int_{D_{\epsilon}} \frac{1}{||t-s||^{N-\alpha}}dt ds = o(\epsilon^N).
\end{split}
\end{equation*}

For $\E[M_u]$, by Kac-Rice formula in \Citet{Adler2007}
\begin{equation*}
 \E[M_u] = \int_{D_{\epsilon}} p_{\nabla X(s)}(0) \E\left[|{\rm det} \nabla^2 X(s)|\mathbbm{1}_{\{\nabla^2 X(s) \prec 0\}}\mathbbm{1}_{\{X(s)>u\}} | \nabla X(s)=0\right] ds .
\end{equation*}

Denote the integrand by $g(s)$. The function $g(s)$ is continuous and positive over the compact domain $D_{\epsilon}$. Thus $\inf_{s \in D_{\epsilon}} g(s) \geq g_0 > 0$, implying

\begin{equation*}
\E[M_u] \geq g_0  \epsilon^N.
\end{equation*}

\noindent Then \eqref{eqn:small_domain} is an immediate consequence of \eqref{ieq:power_approx}.

For $\E[M_{-\infty}]$, by Kac-Rice formula
\begin{equation*}
 \E[M_{-\infty}] = \int_{D_{\epsilon}} p_{\nabla X(s)}(0) \E\left[|{\rm det} \nabla^2 X(s)|\mathbbm{1}_{\{\nabla^2 X(s) \prec 0\}} | \nabla X(s)=0\right] ds.
\end{equation*}
The integrand is also continuous and positive over the compact domain $D_{\epsilon}$ indicating $\E[M_{-\infty}] = o(1)$ for small $\epsilon$. Thus we have
\begin{equation*}
\E[M_u]_{\adj} = \E[M_u]/\max(1,\E[M_{-\infty}]) = \E[M_u]/\max(1,o(1)) = \E[M_u]
\end{equation*}
for sufficiently small $\epsilon$.
\end{proof}






\subsection{Large threshold}

For large threshold $u$, the following asymptotic result shows power can be precisely approximated by $\E[M_u]$.

\begin{thm}\label{thm:large_u}
For any fixed domain $D$, as $u\to \infty$
\begin{equation}\label{eq:approx_1}
\P[M_u \ge 1] = \E[M_u](1-o(e^{-\alpha u^2}))
\end{equation}
where the error term $o(e^{-\alpha u^2})$ is non-negative and $\alpha>0$ is some constant.
\end{thm}

\begin{proof}

By lemma 3 of \citet{Piterbarg96}, as $u \to \infty$, the factorial moment is super-exponentially small. That means $\exists \alpha >0$ s.t.
\begin{equation*}
\E[M_u(M_u-1)] = o(e^{-\frac{u^2}{2}-\alpha u^2}).
\end{equation*}
Also 
\begin{equation*}
E[M_u] \geq \P[M_u \geq 1] \geq \P[\text{sup} X(s) \geq u] = O(e^{-\frac{u^2}{2}}).
\end{equation*}
Thus, the factorial moment decays exponentially faster than $\E[M_u]$. The result is an immediate consequence of \eqref{ieq:power_approx}.
\end{proof}
Notice that the threshold $u$ does not affect the value of $\E[M_{-\infty}]$ which is part of the adjusted $\E[M_u]$. By \eqref{eq:approx_1}
\begin{equation*}
 \P[M_u \ge 1] = \E[M_u]_{\adj}(1-o(e^{-\alpha u^2}))\max(1,\E[M_{-\infty}]).
 \end{equation*} 
If $\E[M_{-\infty}] > 1$, the adjusted $\E[M_u]$ might be overly conservative for large threshold $u$. Therefore, we only recommend $\E[M_u]$ for this scenario.  

\subsection{Sharp signal}
The following theorem provides an asymptotic power approximation when the signal is sharp. Interestingly, while the power function is generally non-Gaussian, it becomes closer to Gaussian as the signal peaks become sharper.

\begin{thm}\label{thm:sharp_signal}
Let $\theta(s) = ah(s) + \theta_0$ where $h(s)$ is a unimodal mean function with maximum equal to 0 at $s_0$, $a>0$, and $\theta_0$ represents the height. For any fixed threshold $u$, as $a \to \infty$
\begin{equation}
\P[M_u \geq 1]  = \E[M_u] + o(1) = \E[M_u]_{\adj} + o(1) = \Phi(\theta_0-u)(1+o(1)),
\end{equation}
where $\Phi(x)$ is CDF of the standard Gaussian distribution.
\end{thm}

\begin{proof}

By lemma A.1 of \citet{Annals}, as $a \to \infty$ 

\begin{equation*}
\P(M_{-\infty} = 1) \geq 1 - O(\exp(-ca^2)),
\end{equation*}
where $c>0$ is some constant. Therefore $M_{-\infty} \overset{p}{\to} 1$. 

Since $M_u \leq M_{-\infty}$ and both of them only take non-negative integer values, $|M_u(M_u-1)|$ and $|M_{-\infty}(M_{-\infty}-1)|$ are bounded above by $|M(M-1)|$ where $M$ is the number of critical points of the random field $X$. Apply Kac-Rice formula
\begin{equation*}
\begin{split}
\E[M(M-1)] = \int_{D}\int_{D} \E\left[|{\rm det}\nabla^2X(s)||{\rm det}\nabla^2X(t)|\middle|\nabla X(s)=\nabla X(t) = 0 \right]
P_{\nabla X(s),\nabla X(t)}(0,0) ds dt.
\end{split}
\end{equation*}

Denote the integrand by $g(s,t,a)$. The function $g(s,t,a)$ is continuous and positive over the compact domain $D$ and $M(M-1) \overset{p}{\to} 0$ as $a \to \infty$. Thus there exists $g_0>0$ such that $\E[M(M-1)] \leq g_0$. Then by dominated convergence theorem
\begin{equation*}
\E[M_u(M_u-1)] \to 0
\end{equation*}
as $a \to \infty$. Since $M_{-\infty} \overset{p}{\to} 1$, the adjusted $\E[M_u]$
\begin{equation*}
\E[M_u]_{\adj} = \E[M_u]\max(1,\E[M_{-\infty}]) = \E[M_u](1 + o(1)) = \E[M_u] + o(1).
\end{equation*}

To calculate $\E[M_u]$, apply Kac-Rice formula 
\begin{align*} \label{eqn:prop4}
 \E[M_u] =& \int_D p_{\nabla X(s)}(0) \E\left[|{\rm det} \nabla^2 X(s)|\mathbbm{1}_{\{\nabla^2 X(s) \prec 0\}}\mathbbm{1}_{\{X(s)>u\}} | \nabla X(s)=0\right] ds \\
=& \int_D p_{\nabla X(s)}(0) \E\left[|{\rm det} \nabla^2 X(s)|\mathbbm{1}_{\{\nabla^2 X(s) \prec 0\}}\mathbbm{1}_{\{X(s)>u\}} | \nabla X(s)=0\right] ds \\
=& \int_D \frac{1}{(2\pi)^{N/2}\sqrt{\det(\Lambda)}}\exp(-a^2 (\nabla h(s))^T\Lambda^{-1}\nabla h(s)/2)\\
& \E\left[|{\rm det} (\nabla^2 Z(s) + a\nabla^2 h(s))|\mathbbm{1}_{\{\nabla^2 X(s) \prec 0\}}\mathbbm{1}_{\{X(s)>u\}} | \nabla X(s)=0\right] ds
\numberthis
\end{align*}
where $\Lambda$ is the covariance matrix of $\nabla h(s)$. Let $f(s) = (\nabla h(s))^T\Lambda^{-1}\nabla h(s)/2$ which attains its minimum $0$ only at $s_0$. Similar to the proof of A.4 in \citet{Annals}, as $a \to \infty$, \eqref{eqn:prop4} can be approximated by applying Laplace's method
\begin{align*}
\E[M_u] =& \frac{\det(a\nabla^2 h(s_0))}{(2\pi)^{N/2}\sqrt{\det(\Lambda)}} \left(\frac{(2\pi)^N \det(\Lambda)}{a^{2N}\det(\nabla^2 h(s_0))}\right)^{1/2}\Phi(\theta_0-u) + O(a^{-2})\\
=& \Phi(\theta_0-u)+O(a^{-2}).
\end{align*}
This finishes the proof.

\end{proof}

\section{Explicit formulas}\label{sec:explicit}

We have showed that the power for peak detection can be approximated by the expected number of local maxima above $u$, $\E[M_u]$,  under certain scenarios such as small domain and large threshold. Although we can apply the Kac-Rice formula to calculate $\E[M_u]$, it remains difficult to evaluate it explicitly for $N>1$ without making any further assumptions. In this section, we focus on computing $\E[M_u]$ and show a general formula can be obtained if the noise field is isotropic. Furthermore, explicit formulas when $N = 1,2,3$ are derived for application purposes.

\subsection{Isotropic Gaussian fields}\label{sec:iso_Gfield}

Suppose $Z$ is a zero-mean unit-variance isotropic random field. We can write the covariance function of $Z$ as $\E\{Z(s)Z(t)\}=\rho(\|s-t\|^2)$ for an appropriate function $\rho(\cdot): [0,\infty) \rightarrow \R$. Denote
\begin{equation}\label{Eq:kappa}
\rho'=\rho'(0), \quad \rho''=\rho''(0), \quad \kappa=-\rho'/\sqrt{\rho''}.
\end{equation}
where $\rho'$ and $\rho''$ are first and second derivative of function $\rho$ respectively.

The following lemma comes from \citet{Bernoulli}.

\begin{lemma}\label{Lem:cov of isotropic Euclidean} For each $s\in \R^N$ and $i$, $j$, $k$, $l\in\{1,\ldots, N\}$,
	\begin{equation*}
	\begin{split}
	\E\{Z_i(s)Z(s)\}&=\E\{Z_i(s)Z_{jk}(s)\}=0, \\
	\E\{Z_i(s)Z_j(s)\}&=-\E\{Z_{ij}(s)Z(s)\}=-2\rho'\delta_{ij},\\
	\E\{Z_{ij}(s)Z_{kl}(s)\}&=4\rho''(\delta_{ij}\delta_{kl} + \delta_{ik}\delta_{jl} + \delta_{il}\delta_{jk})
	\end{split}
	\end{equation*}
	where $\rho'$ and $\rho''$ are defined in \eqref{Eq:kappa} and $\delta_{ij}$ is the Kronecker delta function.
\end{lemma}
In particular, it follows from Lemma \ref{Lem:cov of isotropic Euclidean} that ${\rm Var}(Z_i(s))=-2\rho'$ and ${\rm Var}(Z_{ii}(s))=12\rho''$ for any $i\in\{1,\ldots, N\}$, implying $\rho'<0$ and $\rho''>0$ and hence $\kappa>0$.

We can use theoretical results from \emph{Gaussian Orthogonally Invariant} (GOI) matrices to make the calculation of $\E{[M_u]}$ easier. GOI matrices were first introduced in \citet{Schwartzman08}, and used for the first time in the context of random fields in \citet{Bernoulli}. It is a class of Gaussian random matrices that are invariant under orthogonal transformations, and can be useful for computing the expected number of critical points of isotropic Gaussian fields. We call an $N\times N$ random matrix $G=(G_{ij})_{1\le i,j\le N}$ GOI with covariance parameter $c$, denoted by $\GOI (c)$, if it is symmetric and all entries are centered Gaussian variables such that
\begin{equation}\label{eq:GOI}
\E[G_{ij}G_{kl}] = \frac{1}{2}(\delta_{ik}\delta_{jl} + \delta_{il}\delta_{jk}) + c\delta_{ij}\delta_{kl}.
\end{equation}

The following lemma is Lemma 3.4 from \citet{Bernoulli}.
\begin{lemma}\label{Lem:GOE for det Hessian} Let the assumptions in Lemma \ref{Lem:cov of isotropic Euclidean} hold. Let  $\widetilde{G}$ and $G$ be $\GOI(1/2)$ and $\GOI((1-\kappa^2)/2)$ matrices respectively. $I_N$ denotes $N\times N$ identity matrix.
	
	(i) The distribution of $\nabla^2 Z(s)$ is the same as that of $\sqrt{8\rho''}\widetilde{G}$.
	
	(ii) The distribution of $(\nabla^2 Z(s)|Z(s)=z)$ is the same as that of $\sqrt{8\rho''}\big[G - \big(\kappa z/\sqrt{2}\big)I_N\big]$.
\end{lemma}

Lemma \ref{Lem:GOE for det Hessian} shows the distribution and conditional distribution of the Hessian matrix of a centered random field $Z(s)$. Next, we establish the corresponding result for non-centered random field $X(s) =  Z(s) + \theta(s)$.

\begin{lemma}\label{Lem:GOI} Let  $\widetilde{G}$ and $G$ be $\GOI(1/2)$ and $\GOI((1-\kappa^2)/2)$ matrices respectively.
	
	(i) The distribution of $\nabla^2 X(s)$ is the same as that of
	\[
	 \sqrt{8\rho''} \widetilde{G}+ \nabla^2 \theta(s).
	\] 
	
	(ii) The distribution of $(\nabla^2 X(s)|X(s)=x)$ is the same as that of 
	\[
	  \sqrt{8\rho''} \left[ G - \frac{\kappa (x-\theta(s))}{\sqrt{2}}I_N \right] + \nabla^2 \theta(s).
	\]
\end{lemma}

\begin{proof} Part (i) is a direct consequence of Lemma \ref{Lem:GOE for det Hessian}. For part (ii), note that $(\nabla^2 X(s)|X(s)=x)$ is equivalent to $( \nabla^2 Z(s)| Z(s)=x-\theta(s)) + \nabla^2 \theta(s)$, and the result follows immediately from Lemma \ref{Lem:GOE for det Hessian}.
\end{proof}

\subsection{General formula under isotropy}

\begin{thm}\label{thm:expected-number}
Let $X(s)= Z(s)+\theta(s)$, where $Z(s)$ is a smooth zero-mean unit-variance isotropic Gaussian random field satisfying Assumption \ref{con:1}, \ref{con:2}. Let $\theta(s)$ a smooth $C^3$ mean function such that $\nabla^2\theta(s)$ is a non-singular matrix with ordered eigenvalues $\theta''_1(s)...\theta''_N(s)$ at all critical points $s$. Then for any domain $D$
\begin{equation}
\E[M_u] = \bigg(\frac{2\rho''}{-\pi\rho'}\bigg)^{N/2} \int_D  e^{\frac{\|\nabla \theta(s)\|^2}{4\rho'}} \int_u^{\infty} \phi\left({x-\theta(s)}{}\right)  \E\left[|{\rm det} ({\rm Matrix}(s))|\mathbbm{1}_{\{{\rm Matrix}(s) \prec 0\}}\right] dx \, ds
\label{eqn:M_u_explicit}
\end{equation}
\noindent where $\phi(x)$ is the PDF of the standard Gaussian distribution, ${\rm Matrix}(s)=G - \kappa (x-\theta(s)) I_N/\sqrt{2} + \\$$\diag\{\theta''_1(s),\dots,\theta''_N(s)\} / \sqrt{8\rho''}$, $G$ as in Lemma \ref{Lem:GOI} represents GOI((1-$\kappa^2$)/2), and $\mathbbm{1_{\{\cdot\}}}$ denotes the indicator function.
\end{thm}

\begin{proof}
By the Kac-Rice formula
\begin{align*}
 \E[M_u] &= \int_D p_{\nabla X(s)}(0) \E\left[|{\rm det} \nabla^2 X(s)|\mathbbm{1}_{\{\nabla^2 X(s) \prec 0\}}\mathbbm{1}_{\{X(s)>u\}} | \nabla X(s)=0\right] ds \\
& = \int_D p_{ \nabla Z(s)+\nabla \theta(s)}(0) \E\left[|{\rm det} \nabla^2 X(s)|\mathbbm{1}_{\{\nabla^2 X(s) \prec 0\}}\mathbbm{1}_{\{X(s)>u\}} | \nabla X(s)=0\right] ds \\
& = \int_D \frac{1}{(2\pi)^{N/2}(-2\rho')^{N/2}} e^{\frac{\|\nabla \theta(s)\|^2}{4\rho'}} \E\left[|{\rm det} \nabla^2 X(s)|\mathbbm{1}_{\{\nabla^2 X(s) \prec 0\}}\mathbbm{1}_{\{X(s)>u\}} | \nabla X(s)=0\right] ds \\
& = \int_D   \frac{(8\rho''^2)^{N/2}}{(2\pi)^{N/2}(-2\rho')^{N/2}} e^{\frac{\|\nabla \theta(s)\|^2}{4\rho'}} \int_u^{\infty}   \phi\left({x-\theta(s)}\right)  \E\left[|{\rm det} ({\rm Matrix}(s))|\mathbbm{1}_{\{{\rm Matrix}(s) \prec 0\}}\right] dx \, ds \\
& = \int_{D} \bigg(\frac{2\rho'' }{-\pi\rho'}\bigg)^{N/2} e^{\frac{\|\nabla \theta(s)\|^2}{4\rho'}} \int_u^{\infty}   \phi\left({x-\theta(s)}\right)  \E\left[|{\rm det} ({\rm Matrix}(s))|\mathbbm{1}_{\{{\rm Matrix}(s) \prec 0\}}\right] dx \, ds\\
& = \bigg(\frac{2\rho''}{-\pi\rho'}\bigg)^{N/2} \int_{D} e^{\frac{\|\nabla \theta(s)\|^2}{4\rho'}} \int_u^{\infty}   \phi\left({x-\theta(s)}\right)  \E\left[|{\rm det} ({\rm Matrix}(s))|\mathbbm{1}_{\{{\rm Matrix}(s) \prec 0\}}\right] dx \, ds \\
\end{align*}
Next, we show the derivation from the third to the fourth line in the equation above. Since we assume $\nabla^2 \theta(s)$ is a non-singular matrix at all critical points, then there exists an orthonormal matrix, denoted by $A(s)$, such that $A(s)^T\nabla^2 \theta(s)A(s) = \diag\{\theta''_1(s), \theta''_2, \ldots, \theta''_N(s)\}$, where $\theta''_1\le \ldots \le \theta''_N(s)$ are ordered eigenvalues of $\nabla^2 \theta(s)$. On the other hand, GOI matrices are invariant under orthonormal transformations. By Lemma \ref{Lem:GOI}, the conditional expectation $\E[|{\rm det}(\nabla^2 X(s))|\mathbbm{1}_{\{\nabla^2 X(s)\prec 0\}} | X(s)=x]$ is therefore 
\begin{align*}
& = \E\left[\left|{\rm det}\left(\sqrt{8\rho''} \left[G - \frac{\kappa (x-\theta(s))}{\sqrt{2} }I_N\right] + \nabla^2 \theta(s)\right)\right|\mathbbm{1}_{\{{\rm Matrix}(s) \prec 0\}} \right] \\
& = \E\left[\left|{\rm det}\left(\sqrt{8\rho''} \left[G -\frac{\kappa (x-\theta(s))}{\sqrt{2} } I_N\right] + A(s)^T\nabla^2 \theta(s)A(s)\right)\right|\mathbbm{1}_{\{{\rm Matrix}(s) \prec 0\}} \right] \\
& =  (\sqrt{8\rho''} )^N \E\left[\left|{\rm det}\left(\left[G - \frac{\kappa (x-\theta(s))}{\sqrt{2} } I_N\right] + A(s)^T\nabla^2 \theta(s)A(s) / \sqrt{8\rho''} \right)\right|\mathbbm{1}_{\{{\rm Matrix}(s) \prec 0\}} \right] \\
& =  (\sqrt{8\rho''} )^N \E\left[\left|{\rm det}\left(G - \frac{\kappa (x-\theta(s))}{\sqrt{2} } I_N + \diag\{\theta''_1(s), \theta''_2, \ldots, \theta''_N(s)\} / \sqrt{8\rho''} \right)\right|\mathbbm{1}_{\{{\rm Matrix}(s) \prec 0\}} \right].
\numberthis
\label{eqn:E_Mu}
\end{align*}

\end{proof}

The expression \eqref{eqn:M_u_explicit} can be simplified further if we further assume the mean function $\theta(s)$ to be a rotational symmetric paraboloid centered at $s_0$. In this case, the Hessian of $\theta(s)$ is the identity matrix multiplied by a constant, i.e.
\begin{equation*}
\label{eqn:constant_hessian}
\theta'' = \theta''_1(s) = \theta''_2(s) = ... = \theta''_N(s).
\end{equation*}
Then we can write the mean function as $\theta(s) = \theta_0 + \theta''\|s-s_0\|^2/2 $. 
Define
\begin{equation}\label{eqn:eta}
\eta = \frac{\theta''}{2\kappa \sqrt{\rho''}} = \frac{\theta''}{-2\rho'} = \frac{\theta''}{{\rm Var}(Z_1(s))}.
\end{equation}
and
\begin{equation}
H(\tilde{x}) = \E_{\GOI((1-\kappa^2)/2)}^N \left[\prod_{j=1}^N\left|\lambda_j-\frac{\kappa \tilde{x}}{\sqrt{2}} \right|\mathbbm{1}_{\{\lambda_N<\frac{\kappa \tilde{x}}{\sqrt{2}} \}}\right].
\label{eqn:H_tilde}
\end{equation}
$\E[M_u]$ can be simplified as 
\begin{align*}\label{eqn:M_u}
\E[M_u] =& \bigg(\frac{2\rho'' }{-\pi\rho'}\bigg)^{N/2} \int_{D} e^{\frac{\theta''^2\|s-s_0\|^2}{4\rho'}} \int_{\tilde{u}(s)}^{\infty}  \phi\left(\tilde{x}+\eta\right) H(\tilde{x}) d\tilde{x} ds
\numberthis
\end{align*} 
where we make a change of variable $\tilde{x} = x-\theta(s)-\eta$ and $\tilde{u}(s) = u-\theta(s)-\eta$. Note that the parameter $\kappa$ depends on the correlation structure of $Z(s)$.

\subsection{Explicit formulas in 1D, 2D and 3D}

In \eqref{eqn:M_u}, a general formula for $\E[M_u]$ under isotropy was derived. To make the formula easier to apply in practice, we have the following results for computing it in 1D, 2D, and 3D. When $N = 1$, the derivation is simple enough that we do not need additional assumptions on the mean function $\theta(s)$ except those in Theorem \ref{thm:expected-number}, and it follows directly from Kac-Rice formula. When $N = 2$ and $3$, we assume the mean function $\theta(s)$ is a rotational symmetric paraboloid centered at $s_0$. $\E[M_u]$ is calculated by first obtaining explicit formulas for $H(\tilde{x})$, and plugging $H$ into \eqref{eqn:M_u}. 
\begin{prop}\label{prop:1d_explicit}
Let $N = 1$, $X(s)=  Z(s)+\theta(s)$, where $Z(s)$ is a smooth zero-mean unit-variance Gaussian process and $\theta(s)$ is a smooth mean function. Assume additionally that $Z(s)$ is stationary, then
\begin{equation}\label{eqn:M_u_1D}
\E[M_u] = \int_D \frac{\sqrt{-2\rho'(3-\kappa^2)}}{\kappa} \phi\left(\frac{\theta'(s)}{ \sqrt{-2\rho'} }\right)\int_{u}^{\infty} \phi(x-\theta(s)) \psi\left(\frac{\kappa[x-\theta(s)-\eta(s)]}{ \sqrt{3-\kappa^2}}\right) dx \, ds,
\end{equation}
where the function $\psi$ is defined as
\[
\psi(x) = \int_{-\infty}^x \Phi(y)dy = \phi(x) + x\Phi(x), \quad x\in \R.
\]
\end{prop}

\begin{proof}

Since we assume that $Z(s)$ is stationary, $Z'(s)$ is independent of $Z(s)$ and $Z''(s)$, and $\rho'=-{\rm Var}(Z'(s))/2=\E[Z(s)Z''(s)]/2$ and $\rho''={\rm Var}(Z''(s))/12$ do not depend on $s$. Therefore,
\[
{\rm Var}(X(s))=1, \quad {\rm Var}(X'(s))=-{\rm Cov}[X(s)X''(s)]=-2\rho' \quad {\rm and} \quad  {\rm Var}(X''(s))=12\rho''.
\] 
Note that, by the formula of conditional Gaussian distributions, 
\[
X''(s)|X(s)=x \sim N(\theta''(s)+2\rho'(x-\theta(s)), 12\rho''-4\rho'^2).
\] 

By the Kac-Rice formula
\begin{align*}
\E[M_u] & = \int_D p_{X'(s)}(0) \E[|X''(s)|\mathbbm{1}_{\{X(s)> u\}}\mathbbm{1}_{\{X''(s)< 0\}} | X'(s)=0]ds \\
& = \int_D p_{X'(s)}(0) \int_{u}^{\infty} \phi(x-\theta(s)) \int_{-\infty}^0 (-x'')\frac{1}{ \sqrt{12\rho''-4\rho'^2}}\phi\left[\frac{x''-\theta''(s) -2\rho'(x-\theta(s))}{ \sqrt{12\rho''-4\rho'^2}}\right]dx'' \, dx \, ds \\
& = \int_D p_{X'(s)}(0)  \sqrt{12\rho''-4\rho'^2} \int_{u}^{\infty} \phi(x-\theta(s)) \psi\left(\frac{-2\rho'(x-\theta(s))-\theta''(s)}{ \sqrt{12\rho''-4\rho'^2}}\right) dx \, ds \\
& = \int_D \frac{1}{ \sqrt{-2\rho'}} \phi\left(\frac{\theta'(s)}{ \sqrt{-2\rho'} }\right)  \sqrt{12\rho''-4\rho'^2} \int_{u}^{\infty} \phi(x-\theta(s)) \psi\left(\frac{-2\rho'(x-\theta(s))-\theta''(s)}{ \sqrt{12\rho''-4\rho'^2}}\right) dx \, ds \\
& = \int_D \frac{\sqrt{12\rho''-4\rho'^2}}{\sqrt{-2\rho'}} \phi\left(\frac{\theta'(s)}{ \sqrt{-2\rho'} }\right) \int_{u}^{\infty} \phi(x-\theta(s))\psi\left(\frac{-2\rho'(x-\theta(s))-\theta''(s)}{ \sqrt{12\rho''-4\rho'^2}}\right) dx \, ds.
\end{align*}

The second to third line is due to the fact that
\[
\int_{-\infty}^0 (-x) \frac{1}{b} \phi\left(\frac{x+a}{b}\right) dx
= \int_{-\infty}^0 \Phi\left(\frac{x+a}{b}\right) dx
= b \int_{-\infty}^{a/b} \Phi(y) dy
= b \psi\left(\frac{a}{b}\right).
\]

Recall the $\kappa$ \eqref{Eq:kappa} and $\eta$ \eqref{eqn:eta} parameters defined before. We can rewrite $\E[M_u]$ as

\begin{equation*}
\int_D \frac{\sqrt{-2\rho'(3-\kappa^2)}}{\kappa} \phi\left(\frac{\theta'(s)}{ \sqrt{-2\rho'} }\right)\int_{u}^{\infty} \phi(x-\theta(s)) \psi\left(\frac{\kappa[x-\theta(s)-\eta(s)]}{ \sqrt{3-\kappa^2}}\right) dx \, ds.
\end{equation*}
This finishes the proof.
\end{proof}

Note that when $N = 1$,
\begin{equation}
H(\tilde{x}) = \phi\left(\frac{\kappa\tilde{x}}{ \sqrt{3-\kappa^2}}\right) + \frac{\kappa\tilde{x}}{ \sqrt{3-\kappa^2}}\Phi\left(\frac{\kappa\tilde{x}}{ \sqrt{3-\kappa^2}}\right) = \psi\left(\frac{\kappa\tilde{x}}{ \sqrt{3-\kappa^2}}\right) \\
\end{equation}
We need the following lemmas to calculate $H(\tilde{x})$ explicitly when $N = 2$ and $N = 3$. They are direct calculation of integral by parts.

\begin{lemma}\label{Lem:integral} Let $N = 2$, for constant $a > -\frac{1}{2}$ and $b \in \mathbb{R}$ 
\begin{multline}
	\int_{\mathbb{R}^2}\exp\bigg\{-\frac{1}{2}\sum_{i=1}^{2}\lambda_i^2 - \frac{a}{2}\big(\sum_{i=1}^2 \lambda_i \big)^2 \bigg\}\left(\prod_{i=1}^2|\lambda_i-b| \right)|\lambda_1-\lambda_2|\mathbbm{1}_{\{\lambda_1<\lambda_2<b\}}d\lambda_1 \, d\lambda_2 \\ 
	= \frac{\sqrt{2\pi}}{\sqrt{1+a}}e^{-\frac{1+2a}{2(1+a)}b^2}\Phi\left(\frac{1+2a}{\sqrt{1+a}}b\right)+\left(2b^2-\frac{1+4a}{1+2a}\right)\frac{\sqrt{\pi}}{\sqrt{1+2a}}\Phi(\sqrt{2(1+2a)}b)+\frac{b}{1+2a}e^{-(1+2a)b^2}
	\end{multline}
\end{lemma}


\begin{lemma}\label{Lem:integral_3} Let $N = 3$, for constant $a > 0$ and $b \in \mathbb{R}$ 
\begin{align*}
	&\int_{\mathbb{R}^3}\exp\bigg\{-\frac{1}{2}\sum_{i=1}^{3}\lambda_i^2 + \frac{a}{2(2+3a)}\big(\sum_{i=1}^3 \lambda_i \big)^2 \bigg\}\left(\prod_{i=1}^3|\lambda_i-b| \right)\prod_{1\leq i<j\leq 3}|\lambda_i-\lambda_j|\mathbbm{1}_{\{\lambda_1<\lambda_2<\lambda_3<b\}}d\lambda_1 \, d\lambda_2 \, d\lambda_3\\ 
	= &\left[\frac{a^3+6a^2+12a+24}{2(a+2)^2}b^2+\frac{2a^3+3a^2+6a}{4(a+2)}+\frac{3}{2}\right]\frac{1}{\sqrt{\pi(a+2)}}e^{-\frac{b^2}{a+2}}\Phi\left(\frac{2\sqrt{2}b}{\sqrt{(a+2)(3a+2)}}\right) \\
	&+\left[\frac{a+1}{2}b^2+\frac{a^2-a}{2}-1\right]\frac{1}{\sqrt{\pi(a+1)}}e^{-\frac{b^2}{a+1}}\Phi\left(\frac{\sqrt{2}b}{\sqrt{(a+1)(3a+2)}}\right) \\
	&+\left(a+6+\frac{3a^3+12a^2+28a}{2(a+2)}\right)\frac{b}{2\pi(a+2)\sqrt{3a+2}}e^{-\frac{3b^2}{3a+2}} \\
	&+b\left[b^2+\frac{3(a-1)}{2}\right][\Phi_{\Sigma_1}(0,b)+\Phi_{\Sigma_2}(0,b)]
	\end{align*}
where
\begin{equation*}
\begin{split}
 \Sigma_1 = \left(
\begin{array}{cc}
\frac{3}{2} & -1 \\
-1 & \frac{a+2}{2} \\
\end{array}
\right), \quad  \Sigma_2 = \left(
\begin{array}{cc}
\frac{3}{2} & -\frac{1}{2} \\
-\frac{1}{2} & \frac{a+1}{2} \\
\end{array}
\right).
\end{split}
\end{equation*}
\end{lemma}

\begin{prop}\label{prop:2d_explicit}
Let $N = 2$, and assumptions in Theorem \ref{thm:expected-number} hold. Then the function $H$ defined in \eqref{eqn:H_tilde} can be written explicitly as
\begin{align*}
\label{eqn:H_2D}
H(\tilde{x}) = & \frac{\sqrt{2\pi}}{\sqrt{3-\kappa^2}}\phi\left(\frac{\kappa\tilde{x}}{\sqrt{3-\kappa^2}}\right) \Phi\left(\frac{\kappa\tilde{x}}{\sqrt{(2-\kappa^2)(3-\kappa^2)}}\right) + \frac{\kappa^2}{2}(\tilde{x}^2-1)\Phi\left(\frac{\kappa\tilde{x}}{\sqrt{2-\kappa^2}}\right) + \frac{\kappa\sqrt{2-\kappa^2}\tilde{x}}{2} \phi\left(\frac{\kappa\tilde{x}}{\sqrt{2-\kappa^2}}\right)
\numberthis
\end{align*}

\end{prop}

\begin{proof} 

By Lemma \ref{Lem:integral} above with $a = (\kappa^2-1)/(4-2\kappa^2)$ and $b = \kappa\tilde{x}/\sqrt{2}$, and Lemma 2.2 in \cite{Bernoulli} , 

\begin{align*} \label{eq:2d_num}
& \E_{\GOI((1-\kappa^2)/2)}^N \left[\prod_{j=1}^N \left|\lambda_j - \frac{\kappa\tilde{x}}{\sqrt{2}}\right|I\left(\lambda_N< - \frac{\kappa\tilde{x}}{\sqrt{2}}\right)\right] \\
&= \frac{1}{2\sqrt{\pi(2-\kappa^2)}} \int_{\mathbb{R}^2}
\exp\bigg\{-\frac{1}{2}(\lambda_1^2+\lambda_2^2)+\frac{1-\kappa^2}{4(2-\kappa^2)}(\lambda_1+\lambda_2)^2\bigg\}(\lambda_2-\lambda_1) \\
& \quad \left|\lambda_1-\frac{\kappa\tilde{x}}{\sqrt{2}}\right| \left|\lambda_2-\frac{\kappa\tilde{x}}{\sqrt{2}}\right|\mathbbm{1}_{\{\lambda_2<\frac{\kappa\tilde{x}}{\sqrt{2}}\}}d\lambda_1 \, d\lambda_2 \\ 
& = \frac{1}{\sqrt{3-\kappa^2}}e^{-\frac{\kappa^2\tilde{x}^2}{2(3-\kappa^2)}}\Phi\left(\frac{\kappa\tilde{x}}{\sqrt{(2-\kappa^2)(3-\kappa^2)}}\right) 
+ \quad \frac{\kappa^2}{2}\left( \tilde{x}^2-1\right)\Phi\left(\frac{\kappa\tilde{x}}{\sqrt{2-\kappa^2}}\right) 
+ \frac{\kappa\sqrt{2-\kappa^2}\tilde{x}}{2\sqrt{2\pi}}e^{-\frac{\kappa^2\tilde{x}^2}{2(2-\kappa^2)}}. \numberthis
\end{align*}

This simplifies to \eqref{eqn:H_2D}.

\end{proof}

\begin{prop}\label{prop:3d_explicit}
When $N = 3$, let assumptions in Theorem \ref{thm:expected-number} hold. Then the function $H$ defined in \eqref{eqn:H_tilde} can be written explicitly as
\begin{align*}
H(\tilde{x})
& = \Bigg[\frac{\kappa^2\left[(1-\kappa^2)^3+6(1-\kappa^2)^2+12(1-\kappa^2)+24\right]}
{4(3-\kappa^2)^2}\tilde{x}^2 \\
&\quad+ \frac{2(1-\kappa^2)^3+3(1-\kappa^2)^2+6(1-\kappa^2)}{4(3-\kappa^2)} + \frac{3}{2}\Bigg] \frac{\phi(\frac{\kappa \tilde{x}}{\sqrt{3-\kappa^2}})}{\sqrt{\pi(3-\kappa^2)}}\Phi\left(\frac{2\kappa \tilde{x}}{\sqrt{(3-\kappa^2)(5-3\kappa^2)}}\right)\\
&\quad + \left[\frac{\kappa^2(2-\kappa^2)}{4}\tilde{x}^2 - \frac{\kappa^2(1-\kappa^2)}{2} -1\right] \frac{\phi(\frac{\kappa \tilde{x}}{\sqrt{2-\kappa^2}})}{\sqrt{\pi(2-\kappa^2)}}\Phi\left(\frac{\kappa \tilde{x}} {\sqrt{(2-\kappa^2)(5-3\kappa^2)}}\right)\\
&\quad + \left[7-\kappa^2 + \frac{(1-\kappa^2)\left[3(1-\kappa^2)^2 + 12(1-\kappa^2) + 28\right]}{2(3-\kappa^2)} \right] \frac{\kappa \tilde{x}\phi(\sqrt{\frac{3}{5-3\kappa^2}}\kappa\tilde{x})}{2\sqrt{2}\pi(3-\kappa^2)\sqrt{5-3\kappa^2}}\\
&\quad + \frac{\kappa^3}{2\sqrt{2}}\tilde{x}(\tilde{x}^2 - 3)\left[\Phi_{ \Sigma_1}(0,\kappa \tilde{x}/\sqrt{2}) + \Phi_{ \Sigma_2}(0,\kappa \tilde{x}/\sqrt{2})\right],
\end{align*}

where
\begin{equation*}
\begin{split}
 \Sigma_1 = \left(
\begin{array}{cc}
\frac{3}{2} & -1 \\
-1 & \frac{3-\kappa^2}{2} \\
\end{array}
\right), \quad  \Sigma_2 = \left(
\begin{array}{cc}
\frac{3}{2} & -\frac{1}{2} \\
-\frac{1}{2} & \frac{2-\kappa^2}{2} \\
\end{array}
\right).
\end{split}
\end{equation*}
\end{prop}

\begin{proof}
This is a direct result of Lemma \ref{Lem:integral_3} above with $a = 1-\kappa^2$ and $b = \kappa\tilde{x}/\sqrt{2}$.
\end{proof}
Note that for $N = 2$ and $N = 3$, we need to solve an integral over the domain (see \eqref{eqn:M_u}) to get $\E{[M_u]}$. Although we can not derive the explicit form for the entire formula, this can be evaluated in applications with the help of numerical algorithms. 

\subsection{Isotropic unimodal mean function}

We have calculated the explicit formulas assuming the mean function is a concave paraboloid. This is a very strong assumption. However, in a general setting, where the unimodal mean function is rotationally symmetric of any shape, we can apply a multivariate Taylor expansion at the peak and use the second-order approximation to estimate power. For example, suppose the shape of the mean function is proportional to a rotational symmetric Gaussian density
\begin{equation}
\theta(s) = \theta_0\exp\left(-\frac{\|s-s_0\|^2}{2\xi^2} \right)
\label{eqn:gauss_mean}
\end{equation}
where $s_0$ is the center of the mean function and $\xi$ is the signal bandwidth. The Taylor expansion at the center is 
\begin{equation}
\theta(s) =  \theta_0 -\frac{\theta_0 }{2\xi^2}  \|s-s_0\|^2 + o(\|s-s_0\|^2)
\label{eqn:quadratic_approx}
\end{equation}
When the domain size gets small, we neglect the remainder term, and use its quadratic approximation as the mean function. With quadratic mean function, it becomes convenient to use compute $\E{[M_u]}$. We will evaluate the performance of this approach for different domain sizes in the simulation study.

\section{Simulations}\label{sec:sim}

In Section \ref{sec:power_approx} above, we discussed power approximation under different scenarios. We showed the factorial moment $\E[M_u(M_u-1)]$ decays faster than $\E[M_u]$ under some circumstances so that we can use $\E[M_u]$ or adjusted $\E[M_u]$ to approximate power. In this section, a simulation study is conducted to validate each scenario as well as visualize the power function, $\E[M_u]$, and adjusted $\E[M_u]$. Through simulation, we could also get a better sense of applying them to real data.

\subsection{Paraboloidal mean function}

We generate $B=100,000$ centered, unit-variance, smooth isotropic 2D Gaussian random fields over a grid of size $50 \times 50$ pixels as $Z(s)$, each field obtained as the convolution of white Gaussian noise
with a Gaussian kernel of spatial standard deviation 5, and normalized to standard deviation $\sigma = 1$. For the mean function $\mu(s)$, we use a concave paraboloid centered at $s_0=(25,25)$. The equation of the paraboloid is
\begin{equation}
\theta(s) =  \theta_0 - \frac{\|s-s_0\|^2}{2\xi^2} 
\label{eqn:mean_2d}
\end{equation}
where $\xi$ controls the sharpness of the mean function. The smaller $\xi$ is, the sharper the paraboloid will be. $\theta_0$ controls the height of the signal. To maintain the rotational symmetric property of $\theta(s)$, we only consider those circles centered at $s_0$ as domain $D$. The size of $D$ is measured by the radius $\text{Rad}(D)$. The default value of each parameter is listed in Table unless otherwise specified.
\begin{table}[h]
\centering
\begin{tabular}{c
                c
                c
} \toprule
Parameter      &  Default value \\ \toprule
$\text{Rad}(D)$ & 10 \\
$\xi$ & 7\\
$\theta_0$  & 3\\
\bottomrule
\end{tabular} 
\caption{2D simulation: default value of each parameter}
\label{tab:sim_null_short_term}
\end{table} 

The first two panels of Figure \ref{fig:Z&mu} display two instances of $\theta(s)$ and $Z(s)$ respectively. The third panel displays the resulting sum $X(s)$ which is calculated by the signal-plus-noise model. 

In the simulation, we validate and visualize the scenarios presented above, and check the effect of different choices of parameters on the power function, $\E[M_u]$ and adjusted $\E[M_u]$. Four different scenarios are considered as we discussed in Section \ref{sec:power_approx}:
\begin{enumerate}
\item Height equivariance
\item Small domain size Rad(D)
\item Large threshold $u$  
\item Sharp signal (small $\xi$)
\begin{figure}
\centering
\begin{subfigure}[b]{0.3\textwidth}
	\centering
	\includegraphics[scale=0.3,valign=t]{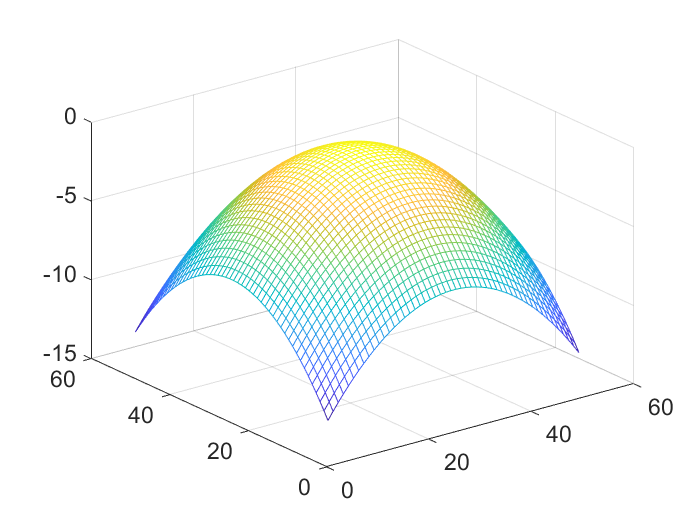}
	\caption{Mean function $\theta(s)$}
	\label{fig:Z&mu_theta}
\end{subfigure}
\begin{subfigure}[b]{0.3\textwidth}
	\centering
	\includegraphics[scale=0.3,valign=t]{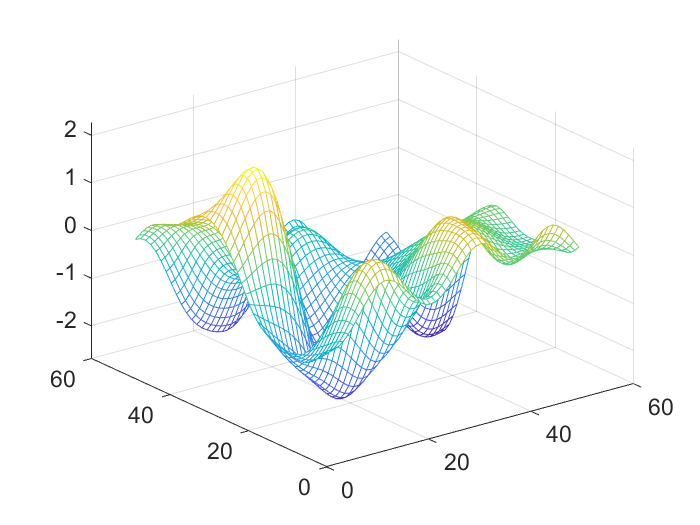}
	\caption{Noise $Z(s)$}
	\label{fig:Z&mu_Z}
\end{subfigure}
\begin{subfigure}[b]{0.3\textwidth}
	\includegraphics[scale=0.3,valign=t]{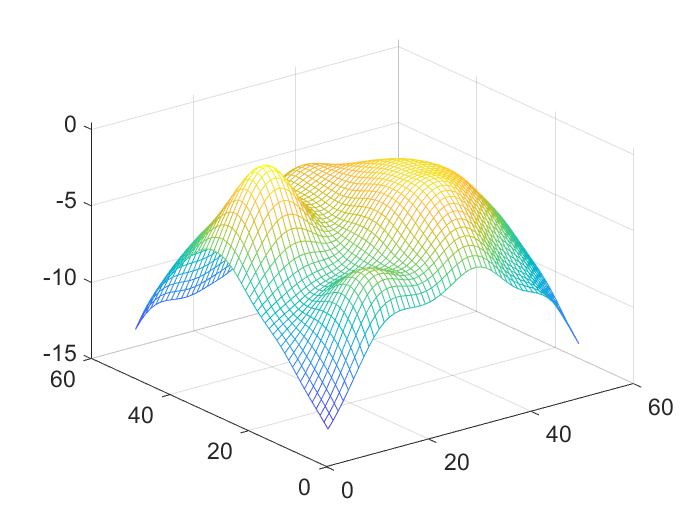}
	\caption{Data $X(s)$}
	\label{fig:Z&mu_X}
\end{subfigure}
\caption{2D simulation: a single instance of $\theta(s)$, $Z(s)$ and their resulting $X(s)$.}
\label{fig:Z&mu}
\end{figure}


\end{enumerate}

 For each simulated random field, we record the height of its highest peak if there exists at least one, and then for any threshold $u$, we calculate the empirical estimate of detection power (\ref{eqn:power_def}) and $\E[M_u]$:
\begin{align}\label{eqn:empirical_power}
 \hat{\P}[M_u \ge 1] &= \frac{1}{B}\sum_{i=1}^B \mathbbm{1}(\exists \text{ a peak in D with height } >u\text{ for $i$th simulated sample}),
\end{align}
\begin{align}\label{eqn:empirical_mu}
\hat{\E}[M_u] &= \frac{1}{B}\sum_{i=1}^B \text{\texttt{\#} peaks in $D$ with height } >u \text{ for $i$th simulated sample}.
\end{align}

Figure \ref{fig:power} displays the power, $\E[M_u]$ and adjusted $\E[M_u]$ curves under the four scenarios. The first panel is to validate scenario 1 (height equivariance). As stipulated by Proposition \ref{prop:height_equi}, the power, $\E[M_u]$ and adjusted $\E[M_u]$ curves are parallel for different signal height $h$ having other parameters remain the same. In the second panel, both the $\E[M_u]$ and adjusted $\E[M_u]$ curve are close to the power curve under scenario 2 (small domain) which indicates that for a smaller domain (quantified by Rad($D$)), using $\E[M_u]$ and adjusted $\E[M_u]$ to approximate power becomes more accurate as stipulated by Theorem \ref{thm:small_domain}. We can also find in all three panels that when $u$ is large enough, the $\E[M_u]$ curve converges to the power curve as $u$ increases, as stipulated by Theorem \ref{thm:large_u}. The adjusted $\E[M_u]$ curve also converges to the power curve but with a slower rate compared to $\E[M_u]$. The third panel shows the power, $\E[M_u]$ and adjusted $\E[M_u]$ curve all converge to the Gaussian CDF for sharp signal (small $\xi$), as stipulated by Theorem \ref{thm:sharp_signal}.

\begin{figure}
\centering
\begin{subfigure}[b]{0.4\textwidth}
	\centering
	\includegraphics[scale=0.35,valign=t]{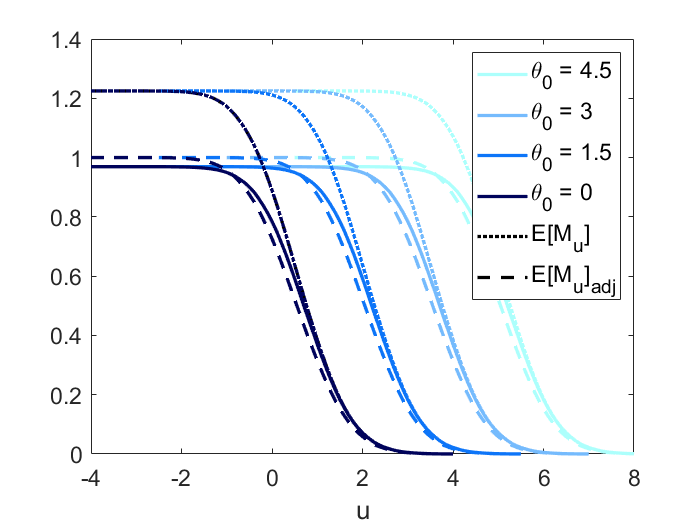}
	\caption{Height equivariance}
	\label{fig:power_height}
\end{subfigure}
\begin{subfigure}[b]{0.4\textwidth}
	\centering
	\includegraphics[scale=0.35,valign=t]{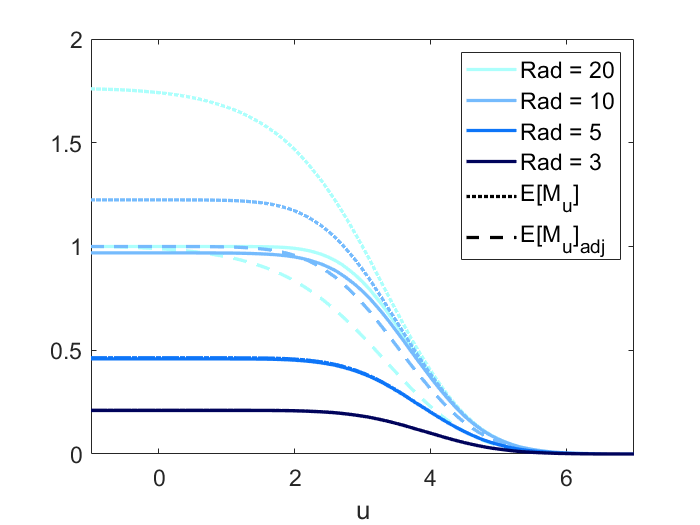}
	\caption{Small domain size}
	\label{fig:power_rad}
\end{subfigure}
\begin{subfigure}[b]{0.4\textwidth}
	\centering
	\includegraphics[scale=0.35,valign=t]{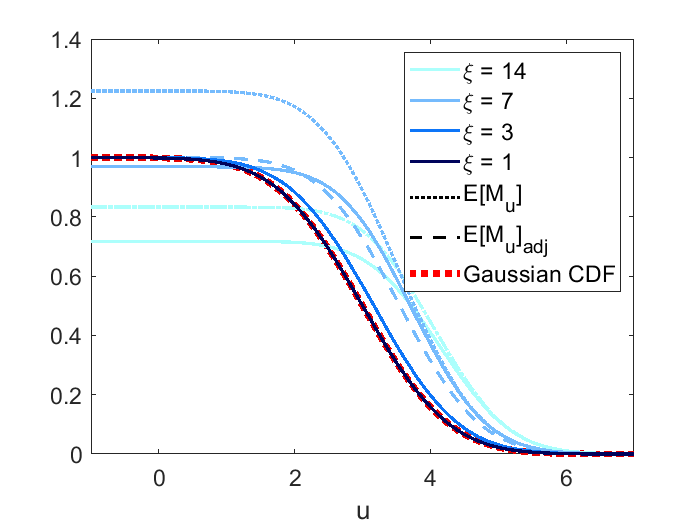}
	\caption{Sharp signal}
	\label{fig:power_xi}
\end{subfigure}
\caption{2D simulation: Power approximation using $E[M_u]$ under four different scenarios (scenario 3 is displayed in all three panels) when the mean function is quadratic.}
\label{fig:power}
\end{figure}

\subsection{Constant mean function}

When the mean function $\theta(s)$ is constant, i.e. it does not depend on location $s$, $X(s)$ reduces to a centered isotropic Gaussian random field. Within the context of this paper, $\theta(s) = 0$ can be seen as the null hypothesis and the power function becomes the probability of Type I error. We use the peak height distribution when $\theta(s) = 0$ (\citealp{Bernoulli}) to decide the cutoff point such that the test meets the nominal type I error. The simulation result when $\theta(s) = 0$ is displayed in Figure \ref{fig:power_constant}.

The performance of Type I error approximation when the mean function is 0 is similar to what we find when the mean function is quadratic (scenario 4 is ignored since the shape parameter does not exist when the mean function is constant). The conclusion is that under large $u$ (which is guaranteed in order to control the Type I error) or small domain, we have good Type I error approximation.

\begin{figure}
\centering
\begin{subfigure}[b]{0.4\textwidth}
	\centering
	\includegraphics[scale=0.35]{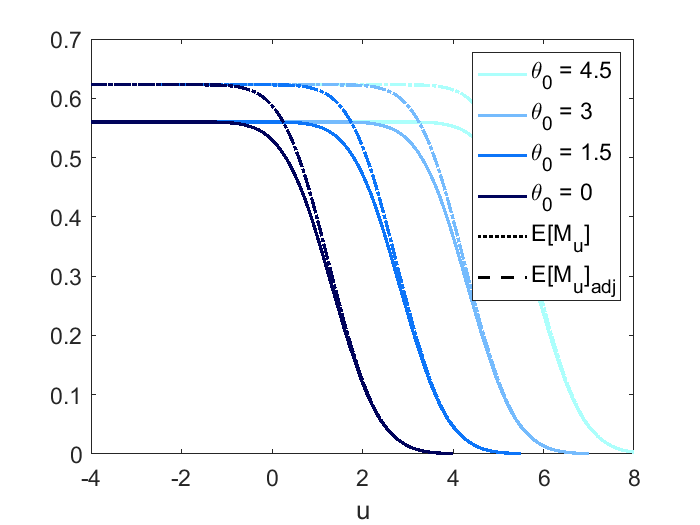}
	\caption{Height equivariance}
	\label{fig:power_constant_height}
\end{subfigure}
\begin{subfigure}[b]{0.4\textwidth}
	\centering
	\includegraphics[scale=0.35]{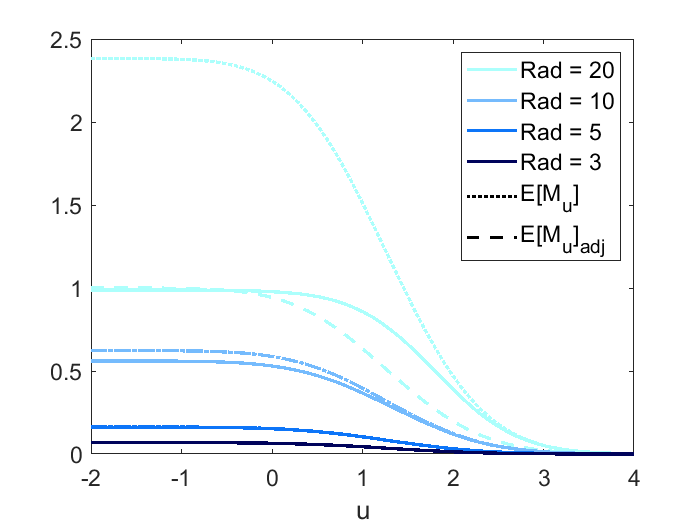}
	\caption{Small domain size}
	\label{fig:power_constant_rad}
\end{subfigure}
\caption{2D simulation: Type I error approximation using $E[M_u]$.}
\label{fig:power_constant}
\end{figure}

\subsection{Gaussian mean function}

The simulation results under Gaussian mean are displayed in Figure \ref{fig:power_gaussian}. For scenario 2 and 3, the results are consistent with those under quadratic mean. For scenario 1, since $\theta_0$ controls both the signal height and sharpness, the power,  $\E[M_u]$ and adjusted $\E[M_u]$ are no longer equivariant in terms of $\theta_0$. For scenario 4, if we look at Figure \ref{fig:power_gaussian_xi} with Figure \ref{fig:power_gaussian_xi_rad_L}, it can be seen that as the signal becomes sharper, the power, $\E[M_u]$ and adjusted $\E[M_u]$ curve converges to the Gaussian CDF only when the domain size (quantified by Rad($D$)) is small. In this case, the asymptotic curve is a mixture of Gaussian CDF and $\E[M_u]$ under constant mean. This is due to the shape of Gaussian density as it converges to 0 if we expand the domain.

In conclusion, for Gaussian mean function, we recommend applying our method to approximate power only when the domain size is small. 

\begin{figure}
\centering
\begin{subfigure}[b]{0.4\textwidth}
	\centering
	\includegraphics[scale=0.35]{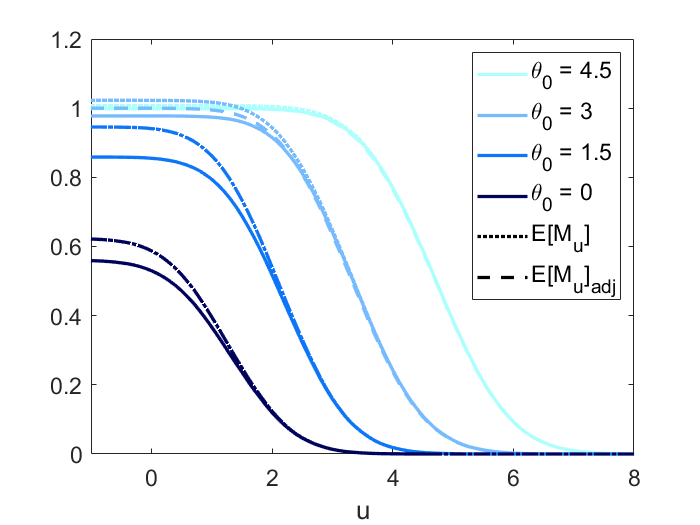}
	\caption{Height equivariance}
	\label{fig:power_gaussian_height}
\end{subfigure}
\begin{subfigure}[b]{0.4\textwidth}
	\centering
	\includegraphics[scale=0.35]{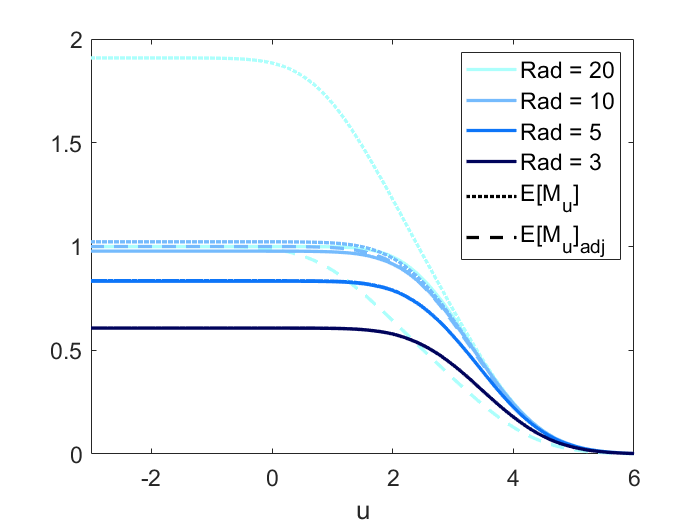}
	\caption{Small domain size}
	\label{fig:power_gaussian_rad}
\end{subfigure}
\begin{subfigure}[b]{0.4\textwidth}
	\centering
	\includegraphics[scale=0.35]{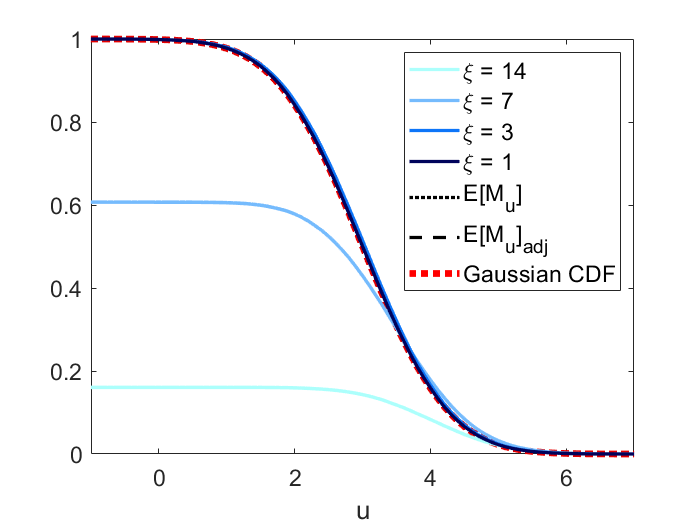}
	\caption{Sharp signal ($\Rad{(D)} = 3$)}
	\label{fig:power_gaussian_xi}
\end{subfigure}
\begin{subfigure}[b]{0.4\textwidth}
    \centering
    \includegraphics[scale=0.35]{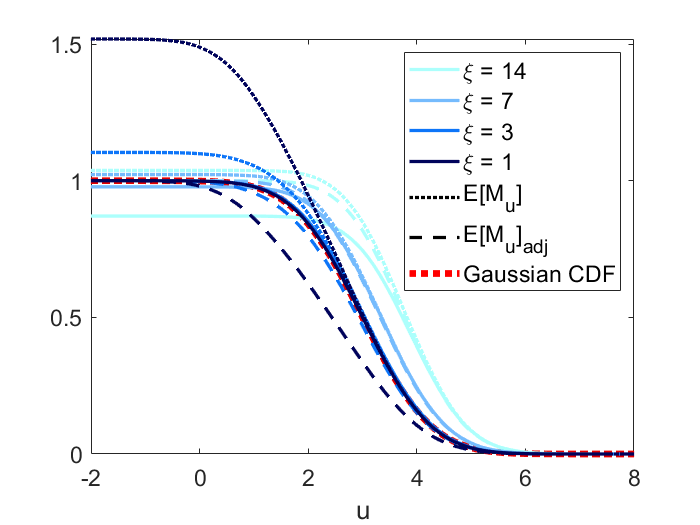}
    \caption{Sharp signal ($\Rad{(D)} = 20$)}
    \label{fig:power_gaussian_xi_rad_L}
\end{subfigure}
\caption{2D simulation: Power approximation using $\E[M_u]$ when the mean function is Gaussian.}
\label{fig:power_gaussian}
\end{figure}


\section{Estimation from data}\label{sec:application}

To use our power approximation formula in real peak detection problems, we need to estimate the spatial covariance function of the noise as well as the mean function from the data. In this section, we demonstrate the 3D application setting and how to estimate the noise spatial covariance function and the mean function. Consider an imaging dataset with $n$ subjects, and let $Y_i(s)$ represent the signal plus noise for subject $i$ 

\begin{equation*}
Y_i(s) = \mu(s) + \sigma(s) \varepsilon_i(s)
\end{equation*}
where $s = (s_1, s_2, s_3)' \in \mathbb{R}^3$, the signal $\mu(s)$, standard deviation $\sigma(s)$ and noise $\varepsilon(s)$ are assumed to be smooth $C^3$ functions. If we compute the standardized mean of all $n$ subjects, we will get a standardized random field 

\begin{equation}
X(s)  = \bar{Y}(s) / \text{SE}(\bar{Y}(s)) = \sqrt{n}\mu(s)/\sigma(s) + \sqrt{n}\bar{\varepsilon}(s)
\label{eqn:T_s}
\end{equation}

This standardized random field $X(s)$ has constant standard deviation of 1. We can treat $\sqrt{n}\mu(s)/\sigma(s)$ as the new signal and $\sqrt{n}\bar{\varepsilon}(s)$ as the new noise of the standardized field. Let $\theta(s) = \sqrt{n}\mu(s)/\sigma(s)$ and $Z(s) = \sqrt{n}\bar{\varepsilon}(s)$. We propose using the following method to estimate the new signal and noise.

\subsection{Estimation of the noise spatial covariance function}\label{sec:kernel_estimate}

We consider the noise $Z(s)$ to be constructed by convolving Gaussian white noise with a kernel:

\begin{equation}
Z(s) = \int_{\mathbb{R}^N} K(t-s) d B(t)
\label{eqn:noise_model}
\end{equation}
where $K(\cdot)$ is a $N$ dimensional kernel function, and $dB(s)$ is Gaussian white noise. Assume that the kernel is rotationally symmetric so that the noise $Z(s)$ is isotropic. Under model (\ref{eqn:noise_model}), we would be able to simulate the noise if we were able to estimate the kernel function from the data.

It can be shown that the autocorrelation of $Z(s)$ is the convolution of the kernel with itself:
\begin{equation*}
\Cor(Z(s),Z(s')) = \int_{\mathbb{R}^N} K(t-s) K(t-s') dt = \int_{\mathbb{R}^N} K(t-(s-s')) K(t) dt.
\end{equation*}

By the convolution theorem, convolution in the original domain equals point-wise multiplication in the Fourier-transformed domain. Thus the kernel function can be estimated empirically using the following method:

\begin{enumerate}
  \item Determine a location $s_0$ of interest (e.g. center of the peak), and calculate the empirical correlation vectors between $Y(s_0)$ and $Y(s)$ where $s$ lies on the three orthogonal axes centered at $s_0$, and belongs to a subdomain of interest.
  \item Take the average of the three estimated correlation vectors (forcing the noise to be isotropic) and perform Fourier transform.
  \item Take the square root of the Fourier coefficients, then the estimated kernel function can be obtained by performing the inverse Fourier transform.
\end{enumerate}

\subsection{Estimation of the mean function}

Our explicit formulas are derived assuming the Hessian of the mean function is a constant times the identity matrix. Therefore, we aim to find a rotational symmetric paraboloid $\hat{\theta}(s)$ that best represents the mean function: 
\begin{equation}
\theta(s) = \beta_0 + \beta_1||s||^2 + (\beta_2, \beta_3, \beta_4)\cdot s
\label{eqn:mean_estimate}
\end{equation}
where the dot represents the vector inner product in $\mathbb{R}^3$. Note that all the quadratic terms share the same coefficient which is due to the rotational symmetry. To estimate (\ref{eqn:mean_estimate}), we can fit a linear regression using all $X(s)$ within the subdomain as outcome.

\section{A 3D real data example}\label{sec:3d_example}

As an application, we illustrate the methods in a group analysis of fMRI data from the Human Connectome Project (HCP) (\Citealp{HCP2012}). The data is used here to get realistic 3D signal and noise parameters from which to do 3D simulations as well as evaluate the performance of our formulas in power approximation. The dataset contains fMRI brain scans for 80 subjects. For each subject, the size of the fMRI image is  $91\times109\times 91$ voxels. The mean and standard deviation of the data are displayed in Figure \ref{fig:hcp_mu_mean} and \ref{fig:hcp_mu_sd}.

\begin{figure}
\centering
\begin{subfigure}[b]{0.4\textwidth}
	\centering
	\includegraphics[scale=0.35]{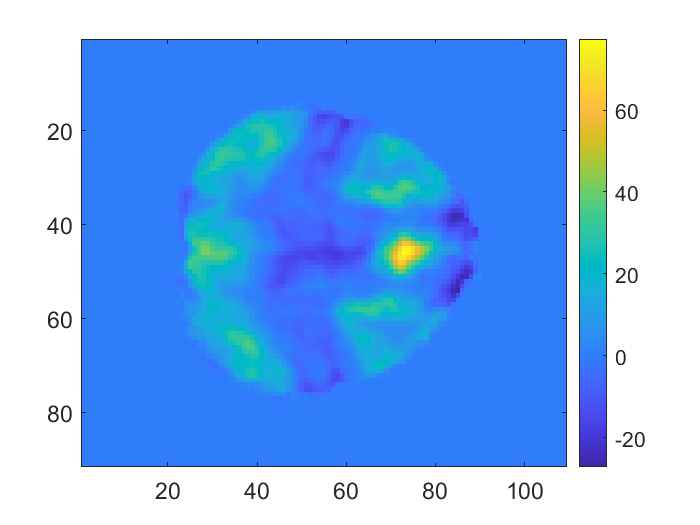}
	\caption{Mean of the data}
	\label{fig:hcp_mu_mean}
\end{subfigure}
\begin{subfigure}[b]{0.4\textwidth}
	\centering
	\includegraphics[scale=0.35]{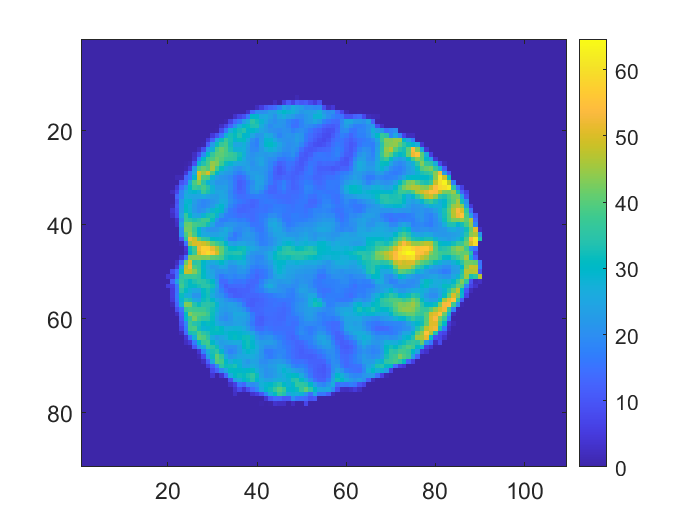}
	\caption{Standard deviation of the data}
	\label{fig:hcp_mu_sd}
\end{subfigure}
\begin{subfigure}[b]{0.4\textwidth}
	\centering
	\includegraphics[scale=0.35]{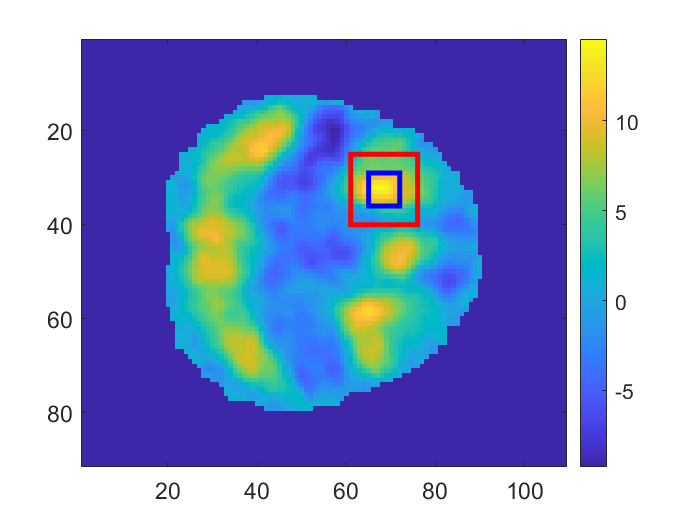}
	\caption{Standardized mean of the smoothed data}
	\label{fig:hcp_mu_mean_smooth}
\end{subfigure}
\caption{HCP data: Mean, standard deviation of the data, and standardized mean of the smoothed data (transverse sliced at the peak of the image along the third dimension). The blue box represents the subdomain of the peak and the red box represents the subdomain we use to estimate the noise spatial covariance function.}
\label{fig:hcp_mu}
\end{figure}

\subsection{Data preprocessing}

Gaussian kernel smoothing is applied to the dataset to make the mean function unimodal around the peak and increase the signal-to-noise ratio. The standard deviation of the smoothing kernel we use in this example is 1 voxel which translates to full width at half maximum (FWHM) being around 2.235. It is obvious from Figure \ref{fig:hcp_mu_sd} that the standard deviation of the noise is not a constant for different locations, thus we use the transformation described in Section \ref{sec:application} to standardize the smoothed data before analyzing it. The standardized mean of the smoothed data $X(s)$ is displayed in Figure \ref{fig:hcp_mu_mean_smooth}.

\subsection{Estimation of the autocorrelation and mean functions}

After standardizing the data, our next step is to estimate the mean and kernel functions using the methods described in Section \ref{sec:application}. Here a $15 \times 15 \times 15$ subdomain (the red box in Figure \ref{fig:hcp_mu_mean_smooth}) around the peak is taken to estimate the kernel. Since we assume the noise is isotropic, the correlation around the peak is supposed to be strictly symmetric along any dimension. However, this is not always true in real data. To tackle this, for each of the three dimensions, we first save the correlation $\Cor(X(s), X(s_0))$ around the peak $s_0$ as a vector and create a new vector by flipping the saved correlation vector. Then we take the mean of the two vectors so that it is guaranteed to be symmetric. The empirical correlation after such symmetrization and the corresponding estimated kernel function are displayed in Figure \ref{fig:hcp_mu_correlation}. 

\begin{figure}
\centering
\includegraphics[scale=0.35]{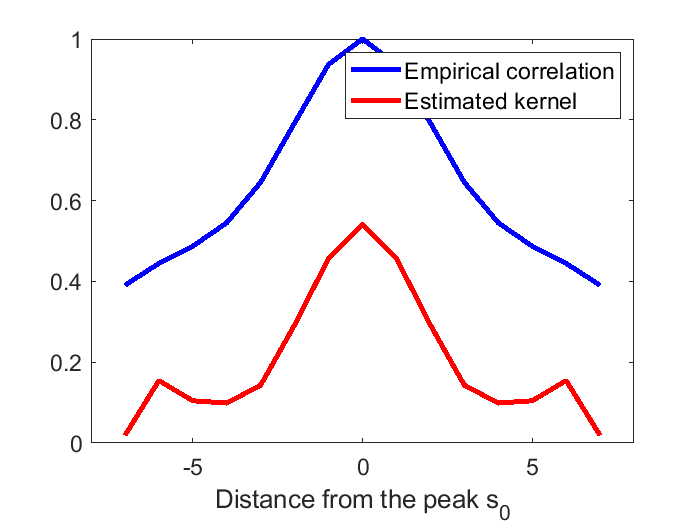}
\caption{The empirical correlation after symmetrization and the estimated kernel from a subdomain of HCP data.}
\label{fig:hcp_mu_correlation}
\end{figure}

We consider two approaches to estimate the mean function, nonparametric and parametric. The nonparametric mean estimation is obtained as a voxelwise average over subjects. 
\begin{equation}\label{eqn:mu_nonpara}
\hat{\theta}(s) = \sum_{i=1}^{n}X_i(s) = \bar{X}(s)
\end{equation}
The parametric mean estimation is obtained by fitting a linear regression model \eqref{eqn:mean_estimate} using all observed data $X(s)$ within the subdomain of size $6 \times 6 \times 6$ (the blue box in Figure \ref{fig:hcp_mu_mean_smooth}) as outcome and their corresponding location variables $||s||^2$, $s$ as covariates. The least square estimate of the mean is
\begin{equation}
\hat{\theta}(s) = 13.03 - 0.26||s||^2 + (0.20, 0.11, 0.39)\cdot s
\label{eqn:mu_hat}
\end{equation}
We will compare the difference in simulated power and $\E[M_u]$ when the mean function is estimated by the nonparametric approach \eqref{eqn:mu_nonpara} vs the parametric approach \eqref{eqn:mu_hat}.

\subsection{3D Simulation induced by data}\label{sec:hcp_sim}

We have done several simulation studies under a well-designed 2D setting where the formulas are supposed to work well, but eventually, we want to apply the formulas to real-life data which is more complicated. Besides, in terms of fMRI data analysis, the image is always 3D by nature. Considering all these factors, a 3D simulation study induced by real data is necessary to validate the performance of the formulas under a more realistic setting, 

In the previous two subsections, we have studied the signal and noise of the HCP data. For the simulation, we would like to generate 3D images using the estimated mean and kernel function. The noise field is generated by convolving the estimated kernel (displayed in Figure \ref{fig:hcp_mu_correlation}) and Gaussian white noise. For each simulation setting, 10,000 such noise fields are generated. 

The signal from the standardized data is very strong (see Figure \ref{fig:hcp_mu_mean_smooth}). For illustrative purposes, we choose to weaken the signal by scaling down the estimated mean function \eqref{eqn:mu_hat} while maintaining the same shape. Here signal strength is measured by effect size, and the amount of scaling is determined by different levels of effect size. In traditional t-test or z-test, Cohen's d values of 0.2, 0.5, and 0.8 (corresponding to 0.58 th 0.69 th and 0.79 th quantiles of the standard Gaussian distribution) are considered as small, medium, and large effect sizes (\citealp{Cohen1988}). The peak height distribution under the null hypothesis (zero mean function) is displayed in Figure \ref{fig:null_distribution}, and does not follow a Gaussian distribution. Therefore, we take 0.58 th, 0.69 th and 0.79 th quantile of the null distribution minus the mean as small (0.16), medium (0.40), and large (0.65) effect size (see the black dash-dot lines in Figure \ref{fig:null_distribution}). For simplicity, we see the peak height of the mean function as effect size in this simulation. However, this is not the most accurate way of defining effect size in the peak detection setting. More details will be discussed in \ref{sec:effect_size}. The threshold $u$ for peak detection is chosen as the 0.99 th quantile of the peak height distribution under the null ($\approx 3.42$) according to \citet{Annals} (see the red dashed line in Figure \ref{fig:null_distribution}). 

\begin{figure}
\centering
\includegraphics[scale=0.35]{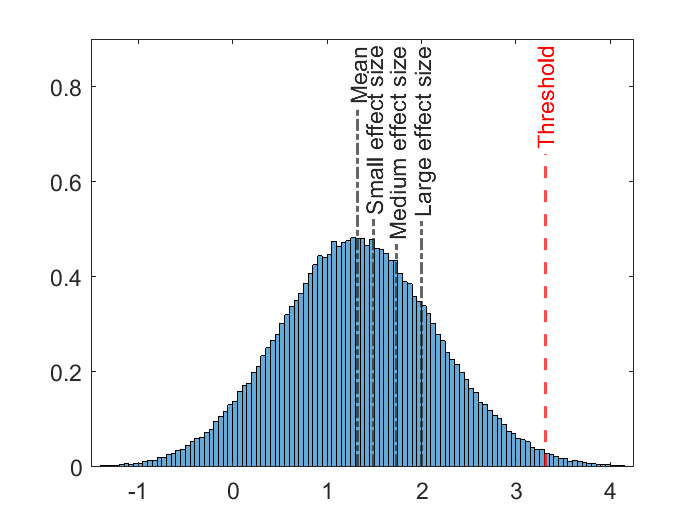}
\caption{3D Simulation induced by data: Simulated Peak height distribution under the null (zero mean) with different levels of effect sizes and threshold $u$.}
\label{fig:null_distribution}
\end{figure}

Similar to the 2D simulation, the search domain $D$ is assumed to be a sphere centered at the true peak, and we use radius of $D$ to control the domain size. Signal sharpness is fixed since it is estimated from the data. The empirical power and $\E[M_u]$ are computed using \eqref{eqn:empirical_power} and \eqref{eqn:empirical_mu}.

To derive the explicit formulas for $\E[M_u]$, we assume the mean function to be a rotational symmetric paraboloid, and this assumption might cause some bias in applications. In Figure \ref{fig:hcp_mu_data}, we demonstrate the difference in simulated power and $\E[M_u]$ when the mean function is estimated by the nonparametric approach \eqref{eqn:mu_nonpara} vs the parametric approach \eqref{eqn:mu_hat}. It can be observed that the quadratic approximation only has a small impact on power and $\E[M_u]$ in this example. In Figure \ref{fig:hcp_power_sim}, we validate our theoretical formula for $\E[M_u]$ \eqref{eqn:M_u} as well as the adjusted $\E[M_u]$. As we can see, the theoretical curve for $\E[M_u]$ and adjusted $\E[M_u]$ closely mirrors the empirical curve. The figure also shows that the power approximation using $\E[M_u]$ is accurate for large $u$ as stipulated by Theorem \ref{thm:large_u}. Power curves using three different effect sizes, and comparisons between large and small domain sizes are displayed in Figure \ref{fig:hcp_power_curve}. We can see from the figure that the approximation works well for small and large sample sizes, and $\E[M_u]_{\adj}$ provides a conservative determination of the sample when $\E[M_u]$ exceeds 1. We can also observe that the performance of power approximation using $\E[M_u]$ becomes better if the domain size is smaller as stipulated by Theorem \ref{thm:small_domain}. 

\begin{figure}
\centering
\includegraphics[scale=0.35]{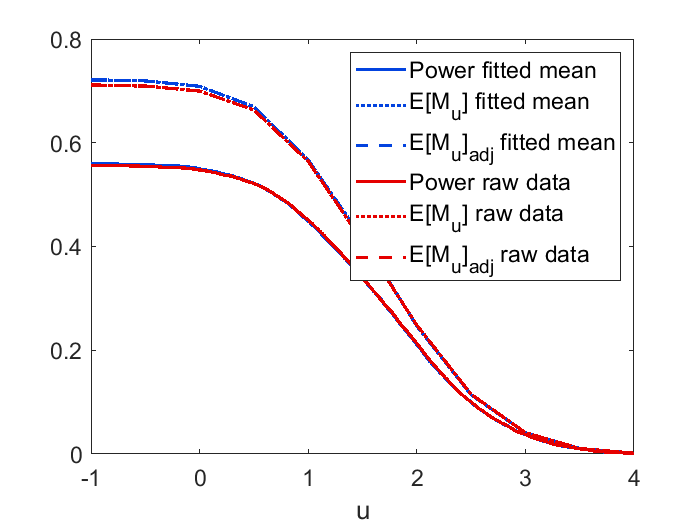}
\caption{3D Simulation induced by data ($\Rad(D) = 3$, medium effect size): Simulated power and $\E[M_u]$ when mean function is obtained from raw data vs quadratic estimation \eqref{eqn:mu_hat}. Here $\E[M_u]$ and adjusted $\E[M_u]$ are the same since $\E[M_{-\infty}] < 1$.}
\label{fig:hcp_mu_data}
\end{figure}

\begin{figure}
\centering
\includegraphics[scale=0.35]{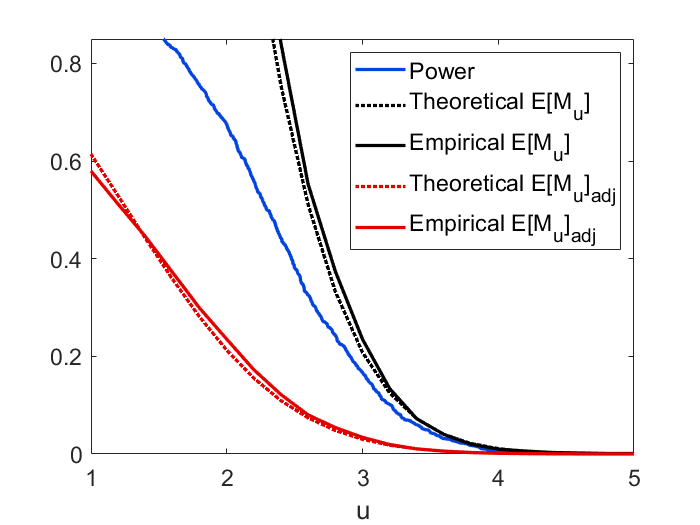}
\caption{3D Simulation induced by data ($\Rad(D) = 6$, medium effect size): Simulated vs theoretical $\E[M_u]$ and adjusted $\E[M_u]$.}
\label{fig:hcp_power_sim}
\end{figure} 



\begin{figure}
\centering
\begin{subfigure}[b]{0.3\textwidth}
	\centering
	\includegraphics[scale=0.29]{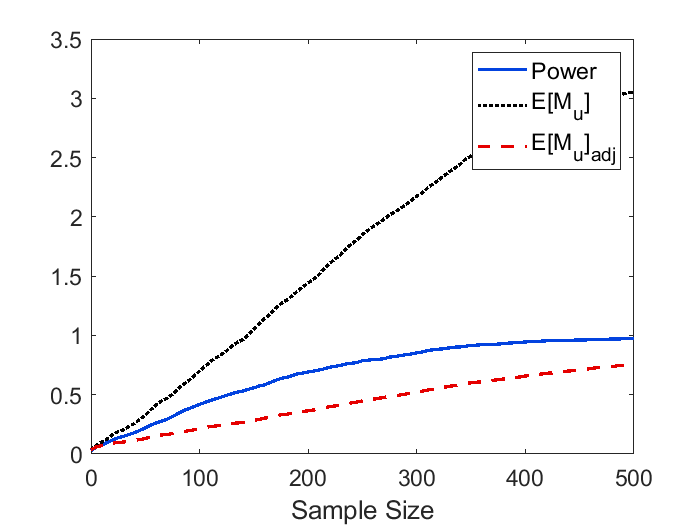}
	\caption{Small effect size (large domain size)}
	\label{fig:hcp_power_curve_small}
\end{subfigure}
\begin{subfigure}[b]{0.3\textwidth}
	\centering
	\includegraphics[scale=0.29]{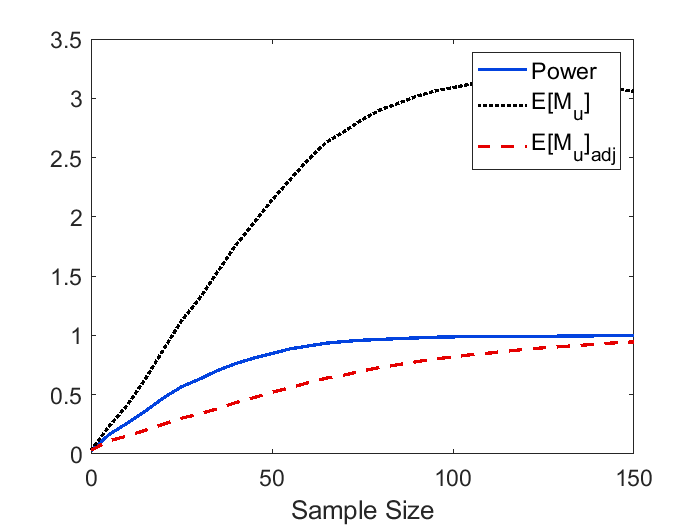}
	\caption{Medium effect size (large domain size)}
	\label{ffig:hcp_power_curve_medium}
\end{subfigure}
\begin{subfigure}[b]{0.3\textwidth}
	\centering
	\includegraphics[scale=0.29]{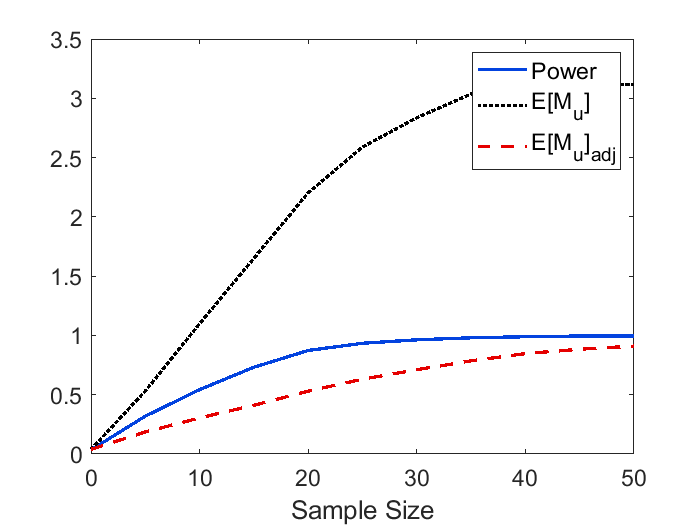}
	\caption{Large effect size (large domain size)}
	\label{fig:hcp_power_curve_large}
\end{subfigure}
\begin{subfigure}[b]{0.3\textwidth}
	\centering
	\includegraphics[scale=0.29]{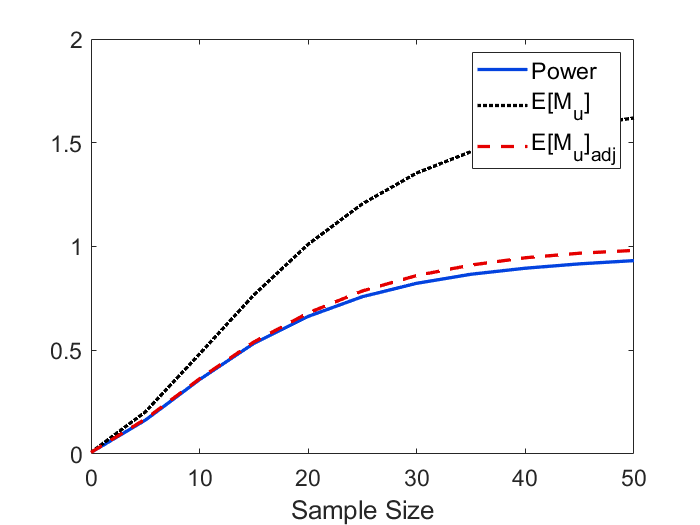}
	\caption{Small effect size (small domain size)}
	\label{fig:hcp_power_curve3_small}
\end{subfigure}
\begin{subfigure}[b]{0.3\textwidth}
	\centering
	\includegraphics[scale=0.29]{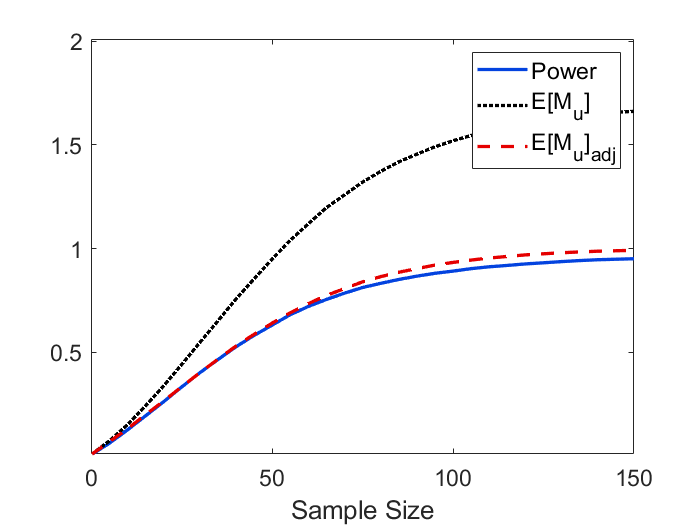}
	\caption{Medium effect size (small domain size)}
	\label{ffig:hcp_power_curve3_medium}
\end{subfigure}
\begin{subfigure}[b]{0.3\textwidth}
	\centering
	\includegraphics[scale=0.29]{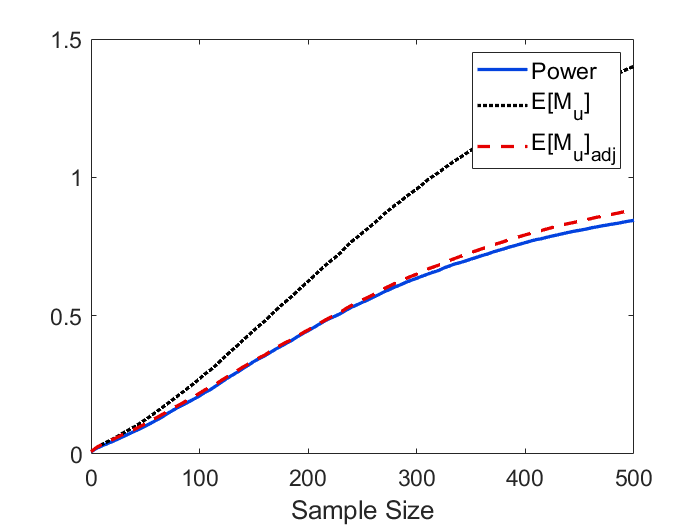}
	\caption{Large effect size (small domain size)}
	\label{fig:hcp_power_curve3_large}
\end{subfigure}
\caption{3D Simulation induced by data: Power curves when the signal has small, medium, and large effect size, and comparisons between large and small domain size.}
\label{fig:hcp_power_curve}
\end{figure}


\section{Discussion}\label{sec:discussion}

\subsection{Explicit formulas and approximations}


Calculating power for peak detection \eqref{eqn:power_def} has been a difficult problem in random field theory due to the lack of formula that can compute it directly. In this paper, we have discussed the rationale of using $\E[M_u]$ and $\E[M_u]_{\adj}$ to approximate peak detection power under different scenarios and derived formulas to compute $\E[M_u]$ assuming isotropy. Isotropy is assumed so that we are able to use the GOI matrix (\citealp{Bernoulli}) as a tool to calculate $\E[M_u]$ via the Kac-Rice formula.

We also showed explicit formulas for $H(\tilde{x})$ (defined as \eqref{eqn:H_tilde}) when $N = 1,2,3$ assuming the mean function is a paraboloid. Computing $H(\tilde{x})$ involves applying the probability density function for the eigenvalues of GOI matrices and details can be found in the proof of Proposition \ref{prop:1d_explicit}, \ref{prop:2d_explicit} and \ref{prop:3d_explicit}. Then $\E[M_u]$ can be calculated by plugging $H(\tilde{x})$ to \eqref{eqn:M_u}. The integration in \eqref{eqn:M_u}, however, can not be evaluated explicitly. In practice, one may evaluate it numerically. For higher dimensions ($N>3$), it remains difficult to get an explicit form of $H(\tilde{x})$ due to the fact inferred by Proposition \ref{prop:1d_explicit}, \ref{prop:2d_explicit} and \ref{prop:3d_explicit} that the integration becomes extremely complicated as $N$ becomes large.


\subsection{Effect size}\label{sec:effect_size}

 We want to emphasize that the power depends on both the signal strength parameter $\theta_0$ and shape parameter $\eta$. In a traditional z-test or t-test which tests a single null hypothesis that the mean value equal to 0, the detection power depends only on a single parameter we call effect size. Here the test is conditional on the point being a local maximum. Applying a simple z-test or t-test, one could reject the null hypothesis as long as the peak height $\theta_0$ exceeds the pre-specified threshold. This approach is not accurate since the peak height does not follow a Gaussian or t distribution. To address this, the threshold can be determined by the null distribution of peak height (\citealp{Bernoulli}) to control the type I error at a nominal level. However, power calculation based on the test over peak height is still biased since the true effect size depends both on the signal height and curvature. The height of the peak affects the likelihood of exceeding the threshold and the curvature affects the likelihood of existing such peak in the domain. It follows that a sharp and high peak is easier to detect compared to a flat and low peak, leading to a larger detection power. 

For an interpretation of the parameter $\eta = \theta''/(-2\rho')$, we consider two types of mean function: paraboloid and Gaussian. Suppose the noise is the result of the convolution of white noise with a  Gaussian kernel with spatial std. dev. $\nu$ resulting in the covariance function with $\rho(r) = \exp(-r/(2\nu^2))$ as specified in Section \ref{sec:iso_Gfield}. This is the same noise as we simulated in Section \ref{sec:application}. When the mean function is paraboloid, consider $\theta(s) = -\|s\|^2/(2\xi^2)+\theta_0$ as in \eqref{eqn:mean_2d}. Here we obtain $\theta'' = -1/\xi^2$ and $\rho' = -1/(2\nu^2)$, yielding $\eta = \theta''/(-2\rho') = -\nu^2/\xi^2$. Thus, $\eta$ is a shape parameter representing the relative sharpness of the mean function with respect to the curvature of the noise. When the mean function is Gaussian, consider $\theta(s) = a \exp(-\|s\|^2/(2\tau^2))$. This expression is obtained, for example, if the signal is the result of the convolution of a delta function with a Gaussian kernel with spatial std. dev. $\tau$. We obtain $\theta'' = -a/\tau^2$ and $\rho' = -1/(2\nu^2)$, yielding $\eta = \theta''/(-2\rho') = -a \nu^2/\tau^2$. Thus, $\eta$ is the height of the signal a, scaled by the ratio of the spatial extent of the noise and signal filters. In both cases, the parameter $\eta$, and thus the power, are invariant under isotropic scaling of the domain, in a similar fashion to the peak height distribution under the null hypothesis (\citealp{CHENG2020108672}).
\begin{figure}
\centering
\begin{subfigure}[b]{0.4\textwidth}
	\centering
	\includegraphics[scale=0.35,valign=t]{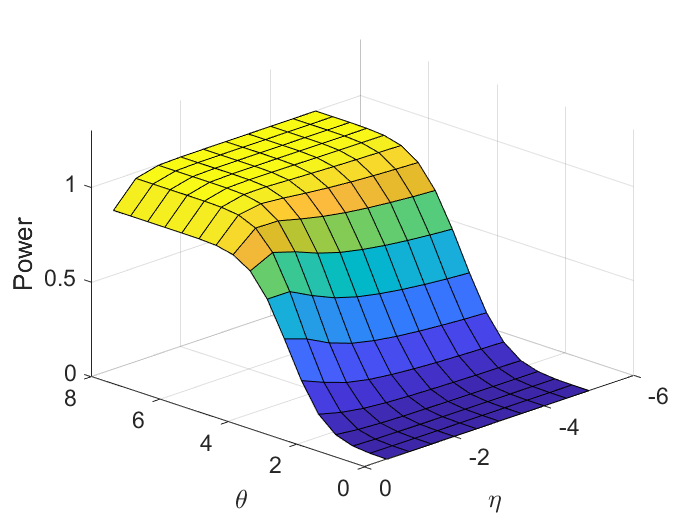}
	\caption{Power}
	\label{fig:Power_eta_theta}
\end{subfigure}
\begin{subfigure}[b]{0.4\textwidth}
	\centering
	\includegraphics[scale=0.35,valign=t]{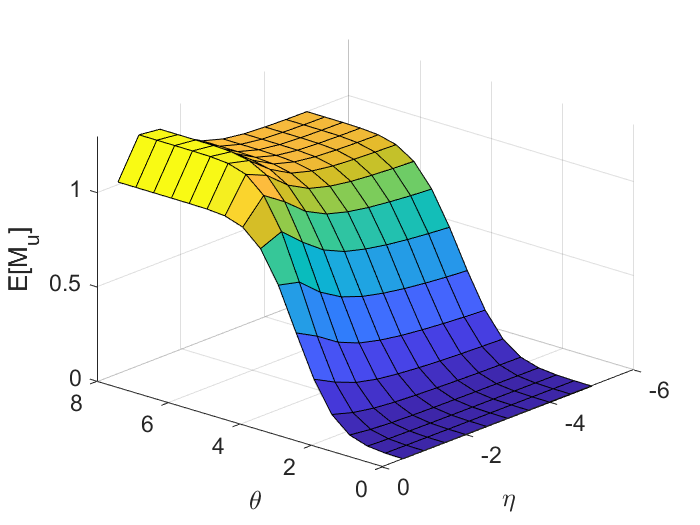}
	\caption{$\E[M_u]$}
	\label{fig:Mu_eta_theta}
\end{subfigure}
\caption{2D simulation: Power and $\E[M_u]$ for different $\theta$ and $\eta$ ($u = 3.92$ and ${\rm Rad}(D) = 10$).}
\label{fig:power_eta_theta}
\end{figure}

Figure \ref{fig:power_eta_theta} illustrates how $\theta_0$ and $\eta$ affect power and $\E[{M_u}]$ in the 2D simulation described in Section \ref{sec:sim}. As we have explained, $\theta_0$ and $\eta$ together determine the effect size. Although deriving an explicit form of effect size as a function of $\theta_0$ and $\eta$ is difficult, we are able to roughly show how the two parameters relate to power. $\theta_0$ which can be seen as signal-to-noise ratio (SNR) plays a major role. Having $\eta$ stay the same, the power monotonically increases with respect to $\theta_0$. On the other hand, power monotonically decreases with respect to $\eta$ having $\theta_0$ stays the same. In this simulation example, the impact of $\theta_0$ on power is about 10 times stronger than $\eta$ if we fit a linear model of power using $\theta_0$ and $\eta$. We can also observe from the figure that the effect of $\eta$ on power is stronger for large $\theta_0$ compared to that for small $\theta_0$.



\subsection{Application to data}


To use our formula to calculate power in practice, one needs to assume the peak to be a certain type such as paraboloid or Gaussian. However, sometimes it might not be plausible to make such assumptions, leading to inaccurate power estimate.

Regarding the conjecture of $\E[M_u]_{\adj}$ being a lower bound when there exists at least one local maximum in the domain $D$, it remains difficult to prove in general, but as we showed in the real data example, it seems to be correct in practice. When it comes to a real-life problem, we can take both $\E[M_u]$ and the $\E[M_u]_{\adj}$ into consideration to get a better understanding of the true sample size. We suggest using $\E{[M_u]}$ as an approximation to power when the sample size is small, considering $\E[M_u]_{\adj}$ when the sample size is large. $\E[M_u]_{\adj}$ also gives a more conservative estimate of power compared to $\E[M_u]$ which is useful to guarantee that the test is powerful enough when we design future studies. Because of its difficulty, we leave further study of $\E[M_u]_{\adj}$ for future work.

\section{Acknowledgments}

Y.Z., D.C. and A.S. were partially supported by NIH grant R01EB026859 and NSF grant 1811659. Data were provided in part by the Human Connectome Project, WU-Minn Consortium (Principal Investigators: David Van Essen and Kamil Ugurbil; 1U54MH091657) funded by the 16 NIH Institutes and Centers that support the NIH Blueprint for Neuroscience Research; and by the McDonnell Center for Systems Neuroscience at Washington University.

\section{Appendix}

\bibliographystyle{myjmva}
\bibliography{ref}

\end{document}